\documentclass[onecolumn,12pt]{IEEEtran}
\usepackage{cite,graphicx,amsmath,amssymb}
\usepackage{citesort}
\usepackage{fancyhdr}
\newtheorem{theorem}{Theorem}

\newtheorem{lemma}{Lemma}

\newtheorem{corollary}{Corollary}

\newtheorem{proposition}{Proposition}

\newtheorem{definition}{Definition}

\begin{document}

\title{Achievable Sum Rate of MIMO MMSE Receivers: A General Analytic Framework}
%\vspace*{0.5cm}
\author{\bigskip Matthew R.\ McKay${}^{\dagger}$, Iain B.\ Collings${}^{*}$, and Antonia M.\ Tulino${}^{\ddagger}$
\\
{\small
${}^\dagger$Department of Electronic and Computer Engineering, \\
        Hong Kong University of Science and Technology, Clear Water Bay, Kowloon, Hong Kong \\
${}^*$Wireless Technologies Laboratory, ICT Centre, CSIRO, Sydney, Australia \\
${}^\ddagger$Department of Electrical Engineering, Universit\'{a} di Napoli ``Federico II'', Napoli, Italy \\
}
}
%\markboth{IEEE Transactions on Information Theory, accepted to appear, 2009}{IEEE Transactions on Information Theory, accepted to appear, 2009}
\maketitle

\setcounter{page}{1}

\vspace*{-1cm}

\begin{abstract}
This paper investigates the achievable sum rate of multiple-input multiple-output (MIMO) wireless systems employing linear minimum mean-squared error (MMSE) receivers. We present a new analytic framework which unveils an interesting connection between the achievable sum rate with MMSE receivers and the ergodic mutual information achieved with optimal receivers.  This simple but powerful result enables the vast prior literature on ergodic MIMO mutual information to be directly applied to the analysis of MMSE receivers. The framework is particularized to various Rayleigh and Rician channel scenarios to yield new exact closed-form expressions for the achievable sum rate, as well as simplified expressions in the asymptotic regimes of high and low signal to noise ratios.  These expressions lead to the discovery of key insights into the performance of MIMO MMSE receivers under practical channel conditions.
\end{abstract}

\vskip 2ex \vspace*{0ex}

\vspace*{1cm}

%{\small
%\begin{tabular}{ll}
%{Corresponding Author\;:} {Matthew R. McKay}\\
%Department of Electronic and Computer Engineering, Hong Kong University of Science and Technology, \\
%{Clear Water Bay, Kowloon, Hong Kong}  \\
%{E-mail: eemckay@ust.hk}, \; \; Phone: +852 2358 5035, \; \; Fax: +852 2358 1485 \\
%\end{tabular}
%}

\newpage

\section{Introduction}

Multiple-input multiple-output (MIMO) antenna wireless communication systems have received enormous attention in recent years due to their ability for providing linear capacity growth without the need for increased power and bandwidth \cite{teletar99,foschini98}. Since the important discoveries in \cite{teletar99,foschini98}, a major focus has been directed at investigating the MIMO channel capacity under a wide range of propagation scenarios. For example, the impact of physical phenomena such as spatial correlation, line-of-sight, antenna mutual coupling, frequency-selectivity, and co-channel interference, have now been well-studied, especially for single-user MIMO systems \cite{shiu00,chuah02,shin03,oyman03,lozano03,smith03,chiani03,alfano04_2,kiessling04,cui05,tulino05_IT,jayaweera05,mckay_jnl05,mck06_let,moustakas06,mckay_07_TrIT}. Many key results have also been derived in the multi-user context, and this is still an important topic of on-going research (see, eg.\ \cite{Blum_02_02,Blum_03_01,Dai_04_01,Chiani_06_Int,Jindal_07_TrIT}, and references therein).

Despite the abundance of literature on MIMO channel capacity, the vast majority of existing work in this area has focused primarily on systems employing optimal nonlinear receiver structures. It is very well-known, however, that such receivers can have prohibitive complexity requirements for practical systems and that low complexity receivers, such as linear minimum mean-squared error (MMSE) receivers, are more appealing for many applications.

However, despite their practical importance, there are currently few closed-form analytical results on the achievable rates of MIMO MMSE receivers.  In fact, most prior related work has focused on studying the \emph{asymptotic} achievable rates of linear MMSE receivers, and has been derived in the context of code-division multiple access (CDMA) systems employing random spreading.  It is well-known that such systems, under certain conditions, are isomorphic to single-user MIMO MMSE systems. In this context, the primary approach has been to study the asymptotic spectral efficiency and signal to interference plus noise ratio (SINR) as the system dimensions grow large with fixed ratio, using advanced tools from large-dimensional random matrix theory \cite{Tse99_IT,shamai01,Zhang01_IT,Verdu02,Guo02_IT,peacock_06_TrIT}.  Those results, which are extremely accurate for the high system dimensions encountered in CDMA applications (determined by the spreading factor and the number of users), may become less accurate for the much smaller system dimensions indicative of current MIMO systems (determined by the numbers of transmit and receive antennas). Moreover, in many cases, the existing CDMA results restrict the equivalent channel gains to be independent random variables with zero mean; thereby precluding many MIMO statistical channel models of practical interest (eg.\ correlated Rayleigh and Rician fading).  An exception is the very recent work \cite{Moustakas_09_VTC}, which investigated the asymptotic (large antenna) mutual information distribution of MIMO MMSE receivers in the presence of spatial correlation.

For finite-dimensional systems with arbitrary numbers of transmit and receive antennas, there are remarkably few analytic results dealing with the achievable rates of MIMO MMSE receivers.  The only directly related results, of which we are aware, were derived very recently in \cite{Kumar_08_TrIT,Louie08,Jindal_07_TrIT}.  Specifically, \cite{Kumar_08_TrIT} investigated the achievable rates of MIMO MMSE receivers by characterizing the asymptotic diversity-multiplexing trade-off. In \cite{Louie08}, expressions for the exact achievable sum rates were derived for MIMO MMSE receivers for uncorrelated Rayleigh fading channels, based on utilizing the distribution of the corresponding SINR at the MMSE receiver output.  A similar method was employed in \cite{Jindal_07_TrIT}, which presented analytic expressions for the high signal to noise ratio (SNR) regime\footnote{Note that \cite{Jindal_07_TrIT} considered the different context of multi-user MIMO broadcast channels with linear zero-forcing precoding; however, for the high SNR regime, there is a strong analogy to the single-user MIMO MMSE model considered in this paper.}, again restricting attention to uncorrelated Rayleigh channels.  The derivation approaches in \cite{Louie08} and \cite{Jindal_07_TrIT} appear intractable for more general channel models. In \cite{Guess_05_TrIT}, properties of mutual information were used to conclude that the MMSE receiver with perfect decision feedback (ie.\ nonlinear) is optimal in certain scenarios.  However they did not consider the realistic case of non-perfect feedback, nor the basic case of no-feedback (ie.\ the linear MMSE receiver); and they did not consider fading.
%Also, they did not provide any actual capacity expressions for the cases considered.

In this paper, we introduce a new general analytic framework for investigating the achievable sum rates of MIMO systems with linear MMSE receivers.  Our main results are based on some very simple but extremely useful algebraic manipulations which essentially relate the MIMO MMSE achievable sum rate to the ergodic MIMO mutual information with optimal receivers.  This relationship permits us to circumvent the extreme difficulties entailed with explicitly characterizing the SINR distribution at the output of the linear MMSE filter, and instead to directly draw upon the vast body of existing results on ergodic mutual information from the MIMO literature. In particular, using this general framework, we can obtain analytic expressions for the achievable sum rate of MIMO MMSE receivers in a broad class of channel scenarios, without the need to invoke a large number of antennas.

We demonstrate our approach by first considering the canonical uncorrelated Rayleigh fading channel, which, as mentioned above, has already been tackled using different methods in \cite{Louie08} and \cite{Jindal_07_TrIT}.  We show that by employing our new framework, it is easy to obtain equivalent expressions to those presented in \cite{Louie08} and \cite{Jindal_07_TrIT}, in addition to establishing new results.  We then consider spatially-correlated Rayleigh and uncorrelated Rician (line-of-sight) fading channels, for which there are no comparable prior results.  For these channels, by employing our general analytic framework, we derive new exact expressions for the achievable sum rates of MIMO MMSE receivers, as well as simplified characterizations in the high and low SNR regimes.

In many cases, our new analytical expressions show a clear decoupling of the effects of transmit correlation, receive correlation, and line-of-sight, which leads to key insights into the performance of MIMO MMSE receivers under practical channel conditions.  For example, they reveal the interesting result that at both high and low SNR, the achievable sum rate of MIMO MMSE receivers is reduced by either spatial correlation or line-of-sight. At high SNR, this rate loss is due to an increased power offset, whereas at low SNR, through a reduced wideband slope. We also present an analytical comparison between the achievable rates of MIMO MMSE receivers, and those of optimal receivers.  Interestingly, we show that at both high and low SNRs, although both MMSE and optimal receivers incur a rate loss due to either spatial correlation or line-of-sight, the loss is more significant for MMSE.

The paper is organized as follows.  In Section \ref{sec:Signal} we introduce the basic MIMO MMSE signal model of interest, and its corresponding achievable rate.  We then present our general analytic framework in Section \ref{sec:GenResults}, before particularizing these general results to various Rayleigh and Rician fading models in Section \ref{sec:Particularizations}.  Finally, Section \ref{sec:conclusions} gives some concluding remarks.

\section{Signal Model and Achievable Sum Rate} \label{sec:Signal}

Consider a single-user MIMO system with $N_t$ transmit and $N_r$
receive antennas, with discrete-time input-output relation
\begin{align}
\mathbf{r} = \mathbf{H} \mathbf{a} + \mathbf{n}
\end{align}
where $\mathbf{r}$ is the $N_r \times 1$ received signal vector, $\mathbf{n}$ is the $N_r \times 1$ vector of additive white Gaussian noise with covariance $ E \left[ \mathbf{n} \mathbf{n}^\dagger \right] = N_0 \mathbf{I}_{N_r} $, and $\mathbf{a}$ is the $N_t \times 1$ vector of transmit symbols, satisfying the total power constraint $E \left[ \mathbf{a}^\dagger \mathbf{a} \right] = P$.  The $N_r \times N_t$ matrix $\mathbf{H}$ represents the flat-fading\footnote{If the fading is frequency-selective, then our results can also be easily applied upon decomposing the channel into a set of parallel non-interacting frequency-flat subchannels.} MIMO channel, assumed to be known perfectly at the receiver but unknown to the transmitter, and normalized to satisfy
\begin{align} \label{eq:ChannelNorm}
E_{\mathbf{H}} \left[ {\rm tr} \left( \mathbf{H} \mathbf{H}^\dagger \right) \right] = N_r N_t \; .
\end{align}
Throughout the paper, we consider the class of MIMO spatial multiplexing systems with independent equal-power Gaussian signalling, in which case the input signals have covariance $E \left[ \mathbf{a} \mathbf{a}^\dagger \right] = \frac{P}{N_t} \mathbf{I}_{N_t}$.

For optimal receivers, the ergodic mutual information\footnote{Note that for the case of uncorrelated Rayleigh fading channels, this corresponds to the ergodic capacity.} is given by
\begin{align}
I^{\rm opt}({\rm snr}, N_r, N_t) = E_{\mathbf{H}}  \left[ I^{\rm opt}({\rm snr}, N_r, N_t, \mathbf{H}) \right]
\label{eq:opt_ExactCap}
\end{align}
where ${\rm snr} = P / N_0$, and
\begin{align}
I^{\rm opt}({\rm snr}, N_r, N_t, \mathbf{H}) = \log_2 {\rm det} \left( \mathbf{I}_{N_r} + \frac{{\rm snr}}{N_t} \mathbf{H} \mathbf{H}^\dagger   \right) \; . \label{eq:CoptDefn_Exp}
\end{align}

In this paper, we focus on characterizing the achievable sum rate of linear MMSE receivers.  Such receivers operate by applying a linear filter to the received signals to form the estimate
\begin{align} \label{eq:MMSEEst}
\mathbf{\hat{a}} = \mathbf{W}_{\rm mmse} \mathbf{r}  \, = \,  \mathbf{W}_{\rm mmse} \left( \mathbf{H} \mathbf{a} + \mathbf{n} \right) \, ,
\end{align}
with $\mathbf{W}_{\rm mmse}$ chosen to minimize the mean-square error cost function
\begin{align}
\mathbf{W}_{\rm mmse} = \arg  \min_{\mathbf{G}} E \left[ \| \mathbf{a} -  \mathbf{G} \mathbf{r} \|^2 \right].
\end{align}
The solution to this optimization problem is well-known (see, eg.\ \cite{McDon95}), and is given by
\begin{align}
\mathbf{W}_{\rm mmse} &= \sqrt{ \frac{N_t}{P} } \mathbf{H}^\dagger  \left[ \mathbf{H} \mathbf{H}^\dagger + \frac{N_t}{{\rm
snr}} \mathbf{I}_{N_r} \right]^{-1} \\
&= \sqrt{ \frac{N_t}{P} } \left[ \mathbf{H}^\dagger \mathbf{H} + \frac{N_t}{{\rm
snr}} \mathbf{I}_{N_t} \right]^{-1} \mathbf{H}^\dagger
\end{align}
where the second line is due to the matrix inversion lemma.
 %Note that the MMSE receiver can be viewed as converting the MIMO channel into a set of $N_t$ parallel (albeit interfering) scalar sub-channels.
It can be easily shown (see eg.\ \cite{Verdu98}) that the instantaneous received SINR for the $i$th filter output (ie.\ corresponding to the $i$th element of $\mathbf{\hat{a}}$) is given by
\begin{align} \label{eq:SINRk}
\gamma_i =  \frac{1}{ \big[ \big( \mathbf{I}_{N_t} + \frac{{\rm
snr}}{N_t} \mathbf{H}^\dagger \mathbf{H} \big)^{-1} \big]_{i,i} } -
1
\end{align}
with $[\cdot]_{i,i}$ denoting the $i$th diagonal element.  Assuming that each filter output is decoded independently, the achievable sum rate is expressed as
\begin{align} \label{eq:MMSEMI}
I^{\rm mmse}({\rm snr}, N_r, N_t) = \sum_{i=1}^{N_t} E_{\gamma_i} \left[ \log_2 \left( 1
+ \gamma_i \right) \right] \; .
\end{align}

In general, the exact distribution of $\gamma_i$ does not appear to be available in closed-form, other than for the specific cases of independent and identically distributed (i.i.d.) or semi-correlated\footnote{The term \emph{semi-correlated} refers to channels with correlation at either the transmitter or receiver, but not both.} Rayleigh fading \cite{gao98}.
%In \cite{kiessling04ISSSTA,zanella05,mckay06}, the moment generating
%function (m.g.f.) of $\gamma_i$ was derived.
This precludes direct evaluation of (\ref{eq:MMSEMI}) for many channels of practical interest.

\section{General Analytic Framework for the Achievable Sum Rate of MIMO MMSE Receivers} \label{sec:GenResults}

In this section, we present our new general analytical framework for investigating the achievable sum rates of MIMO systems with linear MMSE receivers.  In particular, we show that by using some very simple manipulations, the achievable sum rate (\ref{eq:MMSEMI}) can be expressed in a form
which can be easily evaluated for many fading models of interest, without requiring explicit statistical characterization of $\gamma_i$.

We find it convenient to introduce the following notation:
$n = \min (N_r,N_t)$, $m = \max (N_r,N_t)$, $n' = \min(N_r,N_t-1)$, and $m' = \max (N_r,N_t-1)$.

\subsection{Exact Characterization}

The following key result presents a simple and powerful connection between the MIMO MMSE achievable sum rate and the ergodic mutual information obtained with optimal receivers.

\begin{theorem} \label{th:MMSESumCapacity_Main}
The achievable sum rate of MIMO MMSE receivers can be expressed as
\begin{align}
I^{\rm mmse}({\rm snr}, N_r, N_t) &= N_t E_{\mathbf{H}} \left[ I^{\rm opt}({\rm snr}, N_r, N_t, \mathbf{H}) \right] - \sum_{i=1}^{N_t} E_{\mathbf{H}_i} \left[ I^{\rm opt} \left(\frac{N_t-1}{N_t} {\rm snr} , N_r, N_t-1, \mathbf{H}_i \right) \right]
\label{eq:MMSESumCapacity_Main}
\end{align}
where $\mathbf{H}_i$ corresponds to $\mathbf{H}$ with the $i$th
column removed.
\end{theorem}
\begin{proof}
See Appendix \ref{ap:MMSESumCapacity_Main}.
\end{proof}
\begin{corollary}
When $\mathbf{H}$ contains i.i.d.\ entries, (\ref{eq:MMSESumCapacity_Main}) reduces to
\begin{align} \label{eq:IIDRay_MMSE}
I^{\rm mmse}({\rm snr}, N_r, N_t) &= N_t \biggl( I^{\rm opt}({\rm snr}, N_r, N_t) - I^{\rm opt}\left(\frac{N_t-1}{N_t} {\rm snr} , N_r, N_t-1 \right) \biggr) \; .
\end{align}
\end{corollary}

Importantly, with the MMSE achievable sum rate expressed in this form, the
required expectations are the same as those required for the
evaluation of the ergodic MIMO mutual information with optimal receivers which, as already discussed, have well-known solutions for many channels of interest. In Section \ref{sec:Particularizations}  we will draw upon
these previous results to yield new closed-form expressions for the MMSE achievable sum rate.

\subsection{High SNR Characterization}

In the high SNR regime, the ergodic MIMO mutual information and the achievable sum rate of MIMO MMSE receivers can be expressed according to the affine expansion\footnote{The notation $f(x) = o(g(x))$ implies that $\lim_{x \to \infty} \frac{f(x)}{g(x)} = 0$.} \cite{shamai01}
\begin{align} \label{eq:HighSNRGeneral}
I({\rm snr}, N_r, N_t)  = \mathcal{S}_\infty \left( \log_2 {\rm
snr} - \mathcal{L}_\infty \right) + o(1)
\end{align}
where $\mathcal{S}_\infty$ is the high SNR \emph{slope}, in bit/s/Hz/(3 dB) units, given by
\begin{align} \label{eq:slopeDefn}
\mathcal{S}_\infty = \lim_{{\rm snr} \rightarrow \infty}
\frac{ I({\rm snr}, N_r, N_t)}{ \log_2 {\rm snr}}
\end{align}
and $\mathcal{L}_\infty$ is the high SNR \emph{power offset}, in 3
dB units, given by
\begin{align} \label{eq:offsetDefn}
\mathcal{L}_\infty = \lim_{{\rm snr} \rightarrow \infty}
\left( \log_2 {\rm snr} - \frac{ I({\rm snr}, N_r, N_t) }{
\mathcal{S}_\infty } \right) \; .
\end{align}
For MIMO systems with optimal receivers, these parameters are obtained from (\ref{eq:opt_ExactCap}) as
\begin{align}
\mathcal{S}_\infty^{\rm opt} = \min(N_r, N_t)
\end{align}
and
\begin{align} \label{eq:LPO_opt}
\mathcal{L}_\infty^{\rm opt} = \log_2 N_t - \frac{1}{n} E_{\mathbf{H}} \left[ \mathcal{J} ( N_r, N_t, \mathbf{H} ) \right]
\end{align}
respectively, where
\begin{align}
\mathcal{J} ( N_r, N_t, \mathbf{H} ) = \left\{
\begin{array}{ll}
\log_2 {\rm det} \left( \mathbf{H} \mathbf{H}^\dagger \right),   & N_r < N_t \\
\log_2 {\rm det} \left( \mathbf{H}^\dagger \mathbf{H} \right),   & N_r \geq N_t
\end{array}
\right. \; .
\end{align}

For MIMO systems with MMSE receivers, we have the following key result:
\begin{theorem} \label{th:MMSESumCapacity_HighSNR}
At high SNR, the achievable sum rate of MIMO MMSE receivers can be expressed in the general form (\ref{eq:HighSNRGeneral}) with parameters
\begin{align} \label{eq:MMSEHighSNRSlope}
\mathcal{S}_\infty^{\rm mmse} = \left\{
\begin{array}{ll}
N_t , \; \; \; & {\rm for} \; \; N_r \geq N_t \\
0 , \; \; \; & {\rm for} \; \; N_r < N_t
\end{array}
\right. \; \; \;
\end{align}
and $\mathcal{L}_\infty^{\rm mmse}$ is given by
\begin{align}
\mathcal{L}_\infty^{\rm mmse} &= \left\{
\begin{array}{cc}
\log_2 N_t - E_{\mathbf{H}}
\left[ \mathcal{J} ( N_r, N_t, \mathbf{H} )  \right] + \frac{1}{N_t} \sum_{k=1}^{N_t} E_{\mathbf{H}_k} \left[
\mathcal{J} ( N_r, N_t-1, \mathbf{H}_k )  \right] ,  &  N_r \geq N_t  \\
\infty ,  &  N_r < N_t
\end{array}
\right. \; .
\label{eq:MMSE_HighSNR}
\end{align}
\end{theorem}
\begin{proof}
The result is easily obtained by substituting (\ref{eq:MMSESumCapacity_Main}) into (\ref{eq:slopeDefn}) and (\ref{eq:offsetDefn}) and evaluating the necessary limits.
\end{proof}
\begin{corollary} \label{co:MMSESumCapacity_HighSNR}
Consider the case $N_r \geq N_t$.  If $\mathbf{H}$ has i.i.d.\ entries, then (\ref{eq:MMSE_HighSNR}) reduces to
\begin{align}
\mathcal{L}_\infty^{\rm mmse} &= \log_2 N_t - E_{\mathbf{H}}
\left[ \mathcal{J} ( N_r, N_t, \mathbf{H} )  \right]  + E_{\mathbf{H}_1} \left[
\mathcal{J} ( N_r, N_t-1, \mathbf{H}_1 )  \right]  \, .
\label{eq:MMSE_HighSNR_Corr}
\end{align}
\end{corollary}

%Substituting (\ref{eq:MMSE_Capacity_General}) into
%(\ref{eq:slopeDefn}), the high-SNR slope for MIMO MMSE
%systems is obtained as follows
%\begin{align}
%\mathcal{S}_\infty^{\rm mmse} = \left\{
%\begin{array}{ll}
%N_t , \; \; \; & {\rm for} \; \; N_r \geq N_t \\
%0 , \; \; \; & {\rm for} \; \; N_r < N_t
%\end{array}
%\right. \; \; \; .
%\end{align}
The fact that $\mathcal{S}_\infty^{\rm mmse} = 0$ and $\mathcal{L}_\infty^{\rm mmse} = \infty$ for the case $N_r < N_t$ occurs since, when decoding the data stream for a given transmit antenna, the MMSE receiver does not have the capabilities (or enough degrees of freedom) to perfectly cancel the interference caused by the other transmit antennas. Thus, even when the impact of receiver noise becomes negligible (i.e.\ as $\rho \to \infty$), the channel remains interference-limited and, as expected, the MMSE achievable sum rate converges to a finite asymptote.  For the more interesting case, with $N_r \geq N_t$, we see that the high SNR power offset is non-zero.  Very importantly, this case involves the same types of
%given by
%\begin{align}
%\mathcal{L}_\infty^{\rm mmse} &=  \log_2 ( N_t ) - E_{\mathbf{H}}
%\left[ \log_2 \det \left( \mathbf{H}^\dagger \mathbf{H} \right)
%\right] + \frac{1}{N_t} \sum_{k=1}^{N_t} E_{\mathbf{H}} \left[
%\log_2 \det \left( \left(\mathbf{H}^\dagger \mathbf{H}\right)^{kk}
%\right) \right] \nonumber \\
%&=  \log_2 ( N_t ) - E_{\mathbf{H}}
%\left[ \log_2 \det \left( \mathbf{H}^\dagger \mathbf{H} \right)
%\right] + \frac{1}{N_t} \sum_{k=1}^{N_t} E_{\mathbf{H}_k} \left[
%\log_2 \det ( \mathbf{H}_k^\dagger \mathbf{H}_k ) \right] \label{eq:LZF} \; \; .
%\end{align}
expectations as those required for the high SNR analysis of ergodic MIMO mutual information,
for which closed-form solutions exist for a wide range of fading channel scenarios
\cite{grant02,shin03,lozano05_jnl,mckay_jnl05}.  In the sequel, we will draw upon these
previous results in order to derive new simple closed-form expressions for
$\mathcal{L}_\infty^{\rm mmse}$ under a range of conditions.

In addition to the \emph{absolute} high SNR power offset, it is also of interest to
to examine the \emph{excess} high SNR power offset with
respect to the ergodic MIMO mutual information $\mathcal{L}_\infty^{\rm opt}$.  Considering the case $N_r \geq N_t$, this is
given by
\begin{align} \label{eq:DeltaDefn}
\Delta_{\rm ex} = \mathcal{L}_\infty^{\rm mmse} -
\mathcal{L}_\infty^{\rm opt}  \;  .
\end{align}
%From (\ref{eq:MMSE_HighSNR}) and (\ref{eq:LPO_opt}), $\Delta_{\rm ex}$ can be expressed as
%\begin{align}
%\Delta_{\rm ex} &=  \frac{1}{N_t} \sum_{k=1}^{N_t} E_{\mathbf{H}_k} \left[
%\mathcal{J} ( N_r, N_t-1, \mathbf{H}_k )  \right] - \frac{N_t-1}{N_t} E_{\mathbf{H}} \left[ \mathcal{J} ( N_r, N_t, %\mathbf{H} ) \right]  \;
%\end{align}
%and, when $\mathbf{H}$ has i.i.d.\ entries, as simply
%\begin{align}
%\Delta_{\rm ex} &= E_{\mathbf{H}_1} \left[
%\mathcal{J} ( N_r, N_t-1, \mathbf{H}_1 )  \right] - \frac{N_t-1}{N_t} E_{\mathbf{H}} \left[ \mathcal{J} ( N_r, N_t, %\mathbf{H} ) \right]  \; .
%\end{align}

This measure is meaningful, since both the MMSE and optimal receivers yield the same high SNR slope (under the
assumption that $N_r \geq N_t$), and, as such, the corresponding curves will be parallel.
%In fact, this quantity is also powerful since it provides a succinct measure of the sub-optimality of MIMO MMSE receivers, in terms of capacity, in the high SNR regime.
%In the sequel, we will calculate explicit expressions for $\mathcal{L}_\infty^{\rm mmse}$ and $\Delta_{\rm ex}$ under a range of fading channel conditions.

\subsection{Low SNR Characterization}

When considering the low SNR regime, it is convenient to introduce the concept of the \emph{dispersion} of a random matrix.  This measure, originally introduced in \cite{lozano03}, will play a key role in subsequent derivations.

\begin{definition}
Let $\mathbf{\Theta}$ denote a $N \times N$ random matrix.  Then the dispersion of $\mathbf{\Theta}$ is defined as
\begin{align} \label{eq:dispersionDefn}
\zeta (\mathbf{\Theta}) =  N \frac{ E \left[ {\rm tr} ( \mathbf{\Theta}^2 ) \right]}{ E^2 \left[ {\rm tr} ( \mathbf{\Theta} ) \right] } \; .
\end{align}
\end{definition}

For low SNR, it is often appropriate to consider the achievable rate in terms of the normalized transmit energy per information bit, $\frac{E_b}{N_0}$, rather than per-symbol SNR.  This can be obtained from $I ({\rm snr})$ via
\begin{align} \label{eq:LowSNRDefn}
\mathcal{I}\left( {\frac{{E_b }}{{N_0 }}} \right) = I ({\rm snr})
\end{align}
with ${\rm snr}$ the solution to
\begin{align} \label{eq:EbNoNonLinear}
\frac{E_b}{N_0} = \frac{ {\rm snr}}{ I({\rm snr}) } \; .
\end{align}
Note that $\frac{E_b}{N_0}$ is related to the normalized \emph{received} energy per information bit, $\frac{E_b^r}{N_0}$, via
\begin{align}
\frac{E_b^r}{N_0} = N_r \frac{E_b}{N_0} \; .
\end{align}

In general, closed-form analytic formulae for (\ref{eq:LowSNRDefn}) are not forthcoming, however, for low $\frac{E_b}{N_0}$ levels, this  representation is well approximated by \cite{Verdu02}
\begin{align}\label{eq:lowcapa}
\mathcal{I} \left( {\frac{{E_b }}{{N_0 }}} \right) \approx S_0 \log _2
\left( {\frac{{\frac{{E_b }}{{N_0 }}}}{{\frac{{E_b }}{{N_0 }}_{\min
} }}} \right) \;
\end{align}
where the approximation sharpens as $\frac{{E_b }}{{N_0 }} \downarrow \frac{{E_b }}{{N_0 }}_{\rm min}$. Here, ${\frac{{E_b }}{{N_0 }}_{\min } }$ and $S_0$ are the two key parameters which dictate the behavior in the low SNR regime corresponding, respectively, to the minimum normalized energy per information bit required to convey any
positive rate reliably, and the wideband slope. Importantly, they can be calculated directly from $I({\rm snr})$ via \cite{Verdu02}
\begin{align} \label{eq:EbNoDefn}
{\frac{{E_b }}{{N_0 }}_{\min } } &= \lim_{{\rm snr} \to 0} \frac{ {\rm snr} }{ I ( {\rm snr}) } \nonumber \\
&= \frac{1}{\dot{I}(0)}
\end{align}
and
\begin{align} \label{eq:S0Defn}
S_0 &= \lim_{ \frac{E_b}{N_0} \downarrow \frac{E_b}{N_0}_{\min}}
\frac{ \mathcal{I} \left( {\frac{{E_b }}{{N_0 }}} \right) }{ 10 \log_{10} \frac{{E_b }}{{N_0 }} - 10 \log_{10}  \frac{{E_b }}{{N_0 }_{\min} }   } 10 \log_{10} 2  \nonumber \\
&= \frac{-2 ( \dot{I}(0) )^2 }{ \ddot{I}(0) } \ln 2
\end{align}
respectively, where $\dot{I}(\cdot)$ and $\ddot{I}(\cdot)$ denote the first and second-order derivative respectively, taken with respect to ${\rm snr}$. Note that $\mathcal{I} \left( {\frac{{E_b }}{{N_0 }}} \right)$ implicity captures the second-order behavior of $I({\rm snr} )$ as ${\rm snr} \to 0$. For MIMO systems with optimal receivers, these parameters are given by \cite{lozano03}
\begin{align} \label{eq:EbNoMIMO}
\frac{{E_b }}{{N_0 }}_{\min }^{\rm opt} = \frac{\ln 2 }{N_r}
\end{align}
and
\begin{align}
S_0^{\rm opt} = \frac{ 2 N_r }{ \zeta ( \mathbf{H} \mathbf{H}^\dagger ) }
\end{align}
respectively.

For MIMO systems with MMSE receivers, we have the following key result:
\begin{theorem} \label{th:MMSESumCapacity_LowSNR}
At low SNR, the achievable sum rate of MIMO MMSE receivers can be expressed in the general form (\ref{eq:lowcapa}) with parameters
\begin{align} \label{eq:EbNoMMSE}
\frac{{E_b }}{{N_0 }}_{\min }^{\rm mmse} = \frac{\ln 2 }{N_r}
\end{align}
and
\begin{align} \label{eq:WideSlope_MMSE}
S_0^{\rm mmse} = \frac{ 2 N_r }{ N_t \zeta ( \mathbf{H} \mathbf{H}^\dagger ) - \left(\frac{N_t-1}{N_t}\right)^2 \sum_{k=1}^{N_t} \zeta ( \mathbf{H}_k \mathbf{H}_k^\dagger )  } \; \; .
\end{align}
\end{theorem}
\begin{proof}
See Appendix \ref{ap:MMSESumCapacity_LowSNR}.
\end{proof}
\begin{corollary} \label{corr:MMSESumCapacity_LowSNR}
If $\mathbf{H}$ has i.i.d.\ entries, then (\ref{eq:WideSlope_MMSE}) reduces to
\begin{align}
S_0^{\rm mmse} = \frac{ 2 N_r }{ N_t \zeta ( \mathbf{H} \mathbf{H}^\dagger ) - \frac{(N_t-1)^2}{N_t} \zeta ( \mathbf{H}_1 \mathbf{H}_1^\dagger )  } \; \; .   \label{eq:WideSlope_MMSE_IID}
\end{align}
\end{corollary}
\vspace*{0.5cm}
Interestingly, comparison of (\ref{eq:EbNoMIMO}) and (\ref{eq:EbNoMMSE}) reveals that MMSE receivers are optimal in terms of the minimum required $\frac{{E_b }}{{N_0 }}$.  For both receivers, this parameter is independent of the number of transmit antennas, whilst varying inversely with the number of receive antennas; a fact directly attributed to the increased channel energy captured by the additional receive antennas, whilst the total transmit energy is constrained. We also see that the wideband slope of MIMO MMSE receivers depends on the random matrix channel via its dispersion. In the following section we will evaluate this parameter in closed-form for various Rayleigh and Rician fading channels of interest.  From these results, we will see that although MMSE receivers are optimal in terms of the minimum required $\frac{{E_b }}{{N_0 }}$, such receivers are indeed suboptimal in the low SNR regime as typically reflected in a reduced wideband slope $S_0$.

\section{Achievable Sum Rate of MIMO MMSE in Fading Channels} \label{sec:Particularizations}

In this section, we demonstrate the key utility of the general results propounded in the previous section, by presenting explicit solutions for the MIMO MMSE achievable sum rate for various fading models of practical interest.

\subsection{Uncorrelated Rayleigh Fading}

We start by particularizing the results for the canonical case: the i.i.d.\ Rayleigh fading channel,
\begin{align} \label{eq:IIDRay_Model}
\mathbf{H} \sim \mathcal{CN}_{N_r, N_t} \left( \mathbf{0}, \mathbf{I}_{N_r} \otimes \mathbf{I}_{N_t} \right) \; ,
\end{align}
representative of rich scattering non-line-of-sight environments with sufficiently spaced antenna elements.

\begin{figure} \centering
\includegraphics[width=0.7\columnwidth]{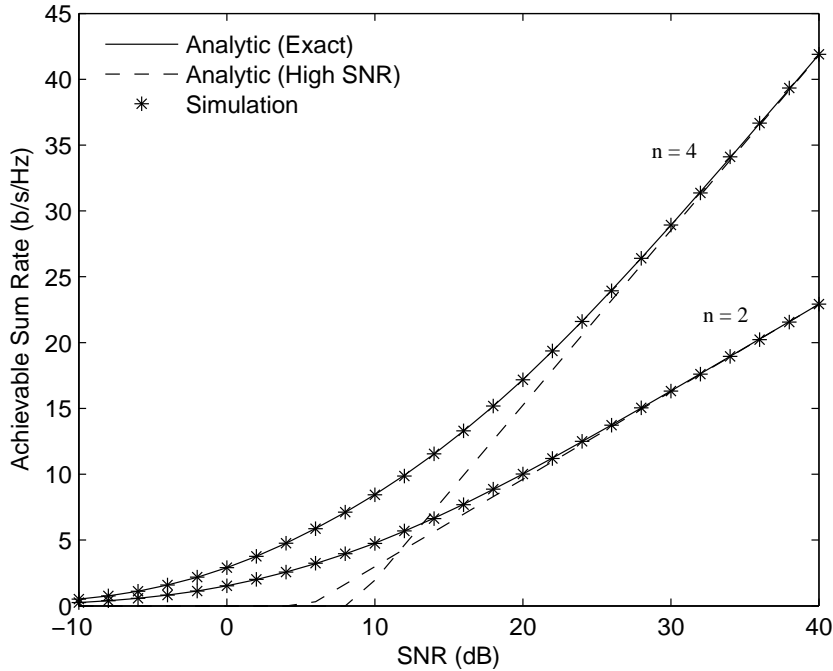}
\caption{Achievable sum rate of MIMO MMSE receivers in i.i.d. Rayleigh fading; comparison of analysis and simulations. Results are shown for different $N_r = N_t = n$.
} \label{fig:MMSE_SumCapacity_IID_Rayleigh}
\end{figure}

\subsubsection{Exact Analysis}

\begin{proposition} \label{th:IIDRayleigh_Exact}
For i.i.d.\ Rayleigh fading, the MIMO MMSE achievable sum rate is given by
\begin{align}
I^{\rm mmse}({\rm snr}, N_r, N_t) = N_t e^{N_t/{\rm snr}}  \left( \frac{
\sum_{k=1}^n \det  \mathbf{\Psi}_{n,m}(k) }{\Gamma_n(m)
\Gamma_n (n)} - \frac{ \sum_{k=1}^{n'} \det \mathbf{\Psi}_{n',m'}(k) }{\Gamma_{n'}(m') \Gamma_{n'} (n')}
 \right) \log_2e \; , \label{eq:IIDRayleigh_Exact}
\end{align}
where $\Gamma_n(\cdot)$ is the normalized complex multivariate gamma
function,
\begin{align}
\Gamma_n (m) = \prod_{i=1}^n \Gamma(m - i + 1)
\end{align}
and $\mathbf{\Psi}_{n,m} (k)$ is an $n \times n$ matrix with
$(s,t)$th element
\begin{align}
\left(\mathbf{\Psi}_{n,m} (k) \right)_{s,t} = \left\{
\begin{array}{ll}
% (n + m - s - t)! \sum_{h=1}^{n+m-s-t+1} {\rm E}_h \left( \frac{N_t}{{\rm snr}}\right) & \text{for} \; \;  t = k \\
% (n + m - s - t)! & \text{for} \; \; t \neq k \\
 \tau_{s,t}! \sum_{h=1}^{\tau_{s,t}+1} {\rm E}_h \left( \frac{N_t}{{\rm snr}}\right) & \text{for} \; \;  t = k \\
 \tau_{s,t}! & \text{for} \; \; t \neq k \\
\end{array}
\right.
\end{align}
where $\tau_{s,t} = n + m - s - t$, and ${\rm E}_h (\cdot)$ is the exponential integral.
\end{proposition}
\begin{proof}
This result is easily obtained by plugging into (\ref{eq:IIDRay_MMSE}) the ergodic mutual information expression for i.i.d.\ Rayleigh MIMO channels given in\footnote{An alternative closed-form expression for ergodic mutual information can be found in \cite{shin03}.} \cite{kang04}.
\end{proof}

Our result in \emph{Proposition \ref{th:IIDRayleigh_Exact}} gives an exact closed-form expression for the MMSE achievable sum rate, which applies for all SNRs and arbitrary antenna configurations.  This result is confirmed in Fig.\ \ref{fig:MMSE_SumCapacity_IID_Rayleigh}, where it is compared with the exact MMSE achievable sum rate, obtained via Monte-Carlo simulations, for different antenna configurations.  There is precise agreement between the simulated and analytic curves, as expected.  We note that \emph{Proposition \ref{th:IIDRayleigh_Exact}} presents a new expression for the achievable sum rate of MIMO MMSE receivers, however, an alternative expression has also been obtained via different means in \cite{Louie08}.  That result was obtained by directly integrating (\ref{eq:MMSEMI}) over the distribution of the SINR in (\ref{eq:SINRk}); an approach that cannot be followed for more general channel models.

For small system dimensions (eg.\ $n = 2$), (\ref{eq:IIDRayleigh_Exact}) reduces to particularly simple forms.  For example, for the case $N_t = 2, N_r \geq 2$, it reduces to
\begin{align}
I^{\rm mmse}({\rm snr}, N_r, 2) &= 2 e^{2/{\rm snr}} \biggl( \sum_{k=1}^{N_r}
{\rm E}_k \left( 2/{\rm snr} \right) + N_r \left( {\rm E}_{N_r + 1} \left( 2/{\rm snr} \right) -  {\rm E}_{N_r} \left( 2/{\rm snr}
\right) \right) \biggr) \log_2e \; ,
\end{align}
whilst for $N_r = 2, N_t \geq 2$, we get
\begin{align}
I^{\rm mmse}({\rm snr, 2, N_t}) &= N_t e^{N_t/{\rm snr}} \biggl( (N_t-1) {\rm
E}_{N_t-1} \left( N_t/{\rm snr} \right) \nonumber \\
& \hspace*{-2cm} + (3 - 2 N_t) {\rm E}_{N_t}
\left( N_t/{\rm snr} \right) + N_t {\rm E}_{N_t+1} \left( N_t/{\rm
snr} \right) \biggr) \log_2e \; .
\end{align}
%For $N_t = N_r = 2$, both results further reduce to
%\begin{align}
%C^{\rm mmse}({\rm snr}, 2, 2) = 2 \log_2 (e) e^{2/{\rm snr}} \left( {\rm E}_1 \left(
%2/{\rm snr} \right) - {\rm E}_2 \left( 2/{\rm snr} \right) + 2 {\rm
%E}_3 \left( 2/{\rm snr} \right) \right) \; .
%\end{align}

\subsubsection{High SNR Analysis}

Recall that in the high SNR regime, the key channel-dependant parameter is the high SNR power offset for the case $N_r \geq N_t$.

\begin{proposition} \label{pr:IIDRayleigh_Proof}
For i.i.d.\ Rayleigh faded channels, the high SNR power offset (for $N_r \geq N_t$) for
MIMO MMSE receivers is given by
\begin{align}
\mathcal{L}_\infty^{\rm mmse} &= \log_2N_t - \log_2e \left(
\sum_{\ell=1}^{N_r - N_t} \frac{1}{\ell} - \gamma \right)
\label{eq:L_ZF_final}
\end{align}
where $\gamma \approx 0.5772$ is the Euler-Mascheroni constant.

For $N_r = N_t = n$ this reduces to
\begin{align}
\mathcal{L}_\infty^{\rm mmse} %= \log_2 (n) + \gamma \log_2 (e)
=\log_2 ( n e^\gamma ) \; \; .
\label{eq:L_ZF_equal}
\end{align}
\end{proposition}
\vspace*{0.5cm}
\begin{proof}
The result is easily obtained from (\ref{eq:MMSE_HighSNR_Corr}), upon noting that \cite{grant02}
\begin{align}
\mathcal{J}(N_r, N_t, \mathbf{H} ) = \log_2e \sum_{\ell = 0}^{N_t - 1} \psi \left(N_r - \ell \right)
\end{align}
for $N_r \geq N_t$,
where $\psi(\cdot)$ is the \emph{digamma} function defined as
\begin{align} \label{eq:Digamma}
\psi(j) = \left\{
\begin{array}{lr}
\sum_{k=1}^{j-1} \frac{1}{k} - \gamma     & {\rm for} \; \, j > 1 \\
-\gamma & {\rm for} \; \, j = 1
\end{array}
\right. \; \; .
\end{align}
%See Appendix \ref{Ap:IIDRayleigh_Proof}.
\end{proof}

Recalling that the MMSE receiver behaves equivalently to the linear zero forcing (ZF) receiver at high SNR, we note that \emph{Proposition \ref{pr:IIDRayleigh_Proof}} could also be easily derived by starting with the high SNR MIMO ZF sum rate expression presented in \cite[Eq. 8.54]{tse05_book} for the case of i.i.d.\ Rayleigh fading channels.

Together with (\ref{eq:MMSEHighSNRSlope}), \emph{Proposition \ref{pr:IIDRayleigh_Proof}} indicates that if the number of transmit antennas is kept fixed and the number of receive antennas are increased, then, whilst having no effect on the high SNR slope, the high SNR achievable rate is improved through a reduction in the power offset. Intuitively, this is due to the additional received power captured by the extra antennas, and also to the enhanced interference cancelation capabilities afforded by the extra degrees of freedom in the receive array.  In fact, as $N_r \gg N_t$, then $\mathcal{L}_\infty^{\rm mmse} \to - \infty$ dB, confirming the intuition that the MMSE receiver completely mitigates the effect of fading (in the high SNR regime) as the number of degrees of freedom at the receiver greatly exceed the number of impeding interferers.

It is also worth noting that, based on (\ref{eq:L_ZF_final}) and (\ref{eq:L_ZF_equal}), one may conclude that increasing the number of transmit and receive antennas, whilst keeping their difference fixed, may have a deleterious effect on the achievable rate due to an increased high SNR power offset; especially when $N_r = N_t = n$.  However, care must be taken when interpreting this result.  In particular, since the high SNR slope (\ref{eq:MMSEHighSNRSlope}) also increases linearly with $N_t$, it turns out that the overall MMSE achievable sum rate actually increases with $n$.  This result is seen in Fig.\ \ref{fig:MMSE_SumCapacity_IID_Rayleigh}, where the high SNR MMSE achievable sum rate based on (\ref{eq:L_ZF_equal}) and (\ref{eq:HighSNRGeneral}) is presented for $n = 2$ and $n = 4$.  We see that the slope is greatest for the case $n = 4$, as expected; however the power offset, which determines the point at which the high SNR linear approximation intersects with the horizontal SNR axis, is smallest for the case $n = 2$.

As an aside, it is also important to note that although the general approximation (\ref{eq:HighSNRGeneral}) is formally valid in the regime of very high SNRs, Fig.\ \ref{fig:MMSE_SumCapacity_IID_Rayleigh} demonstrates good accuracy even for moderate SNR values (eg.\ within $20$ dB).

As the next result shows, the high SNR power offset (\ref{eq:L_ZF_final}) admits further simplifications in the ``large-antenna" regime.

\begin{corollary}
For i.i.d.\ Rayleigh faded MIMO channels, as the number of antennas grows with ratio $\beta = \frac{N_t}{N_r}$  (with $\beta \leq 1$), the high SNR power offset (\ref{eq:L_ZF_final}) converges to
\begin{align}
\mathcal{L}_\infty^{\rm mmse} & \to \log_2 \left( \frac{\beta}{1-\beta} \right)  \; .
\label{eq:L_MMSE_LargeSys}
\end{align}
\begin{proof}
The result is easily established upon noting that
\begin{align} \label{eq:psiAsymptote}
\sum_{\ell=1}^{n-1} \frac{1}{\ell} - \gamma = \psi(n) \sim \ln(n)
\end{align}
for large $n$.
\end{proof}
\end{corollary}

Interestingly, we see that the high SNR power offset is unbounded for $\beta = 1$ (i.e.\ $N_r = N_t$); however, it converges for all $\beta < 1$, decreasing monotonically in $\beta$. We note that this expression agrees with a previous large-system result derived for MMSE receivers in the context of CDMA systems with random spreading \cite{shamai01}.

\begin{corollary} \label{th:ExcessPO_IID}
For i.i.d.\ Rayleigh faded MIMO channels, the excess high SNR power
offset is given by
\begin{align} \label{eq:IIDRayExcessPO}
\Delta_{\rm ex} &= \log_2e \left( \frac{N_r}{N_t} \sum_{\ell =
N_r - N_t + 1}^{N_r} \frac{1}{\ell} - 1 \right) \; .
\end{align}

For $N_r = N_t = n$, this reduces to
\begin{align}
\Delta_{\rm ex} = \log_2e \sum_{\ell = 2}^{n} \frac{1}{\ell} \; .
\end{align}
\end{corollary}
\vspace*{0.5cm}
\begin{proof}
The result is obtained by substituting (\ref{eq:L_ZF_final}) and
\cite[Eq.15]{lozano05_jnl}
\begin{align}
\mathcal{L}_\infty^{\rm opt} &= \log_2 N_t + \log_2e \left( \gamma -
\sum_{\ell=1}^{N_r - N_t} \frac{1}{\ell} - \frac{N_r}{N_t}
\sum_{\ell = N_r-N_t+1}^{N_r} \frac{1}{\ell} + 1 \right)
\end{align}
into (\ref{eq:DeltaDefn}), and performing some basic algebraic manipulations.
%See Appendix \ref{Ap:ExcessPO_IID}.
\end{proof}

Note that an alternative expression for (\ref{eq:IIDRayExcessPO}) can also be obtained from \cite[Theorem 2]{Jindal_07_TrIT} and \cite[Eq. (15)]{Jindal_07_TrIT}, which considered the asymptotic excess rate offset of linear precoding in uncorrelated Rayleigh fading MIMO broadcast channels.

The excess high SNR power offset also admits a simplified characterization in the large-antenna regime.
\begin{corollary} \label{th:IIDRayleigh_LargeSys}
For i.i.d.\ Rayleigh faded MIMO channels, as the number of antennas grows with ratio $\beta = \frac{N_t}{N_r}$  (with $\beta \leq 1$), the excess high SNR power offset (\ref{eq:IIDRayExcessPO}) converges to
\begin{align}
\Delta_{\rm ex} & \to \frac{1}{\beta} \log_2 \left( \frac{1}{1-\beta} \right) - \log_2e \; .
\label{eq:Delta_LargeSys}
\end{align}
\end{corollary}
\begin{proof}
The result is derived trivially from (\ref{eq:IIDRayExcessPO}) upon employing (\ref{eq:psiAsymptote}).
\end{proof}

Again, we note that this expression agrees with a previous large-system result derived in \cite{shamai01}, which considered the context of randomly-spread CDMA systems.

%This shows that for $n \times n$ systems, the power offset can be
%very large.  We can also see, however,  that for this case the
%reduction in SNR offset by the addition of a single receive antenna
%(i.e.\ the case $N_t = n$, $N_r = n+1$) is also very significant,
%given by
%\begin{align}
%3 * ((1 - c) \log_2 (e) - c \log_2 (e)) = 4.33 \; \; {\rm dB} \; .
%\end{align}
%
%Directly comparing the results of Theorem ** with the power offset
%for the i.i.d.\ Rayleigh channels with optimal detection, given in
%\cite[LOZ, Prop.\ 1]{lozano05_jnl}, the excess power offset (in dB)
%incurred by using the sub-optimal MMSE detector is quantified as
%\begin{align}
%\Delta_{\mathcal{L}_\infty} = 3 \log_2 (e) \left( \frac{N_r}{N_t}
%\sum_{\ell = N_r - N_t + 1}^{N_r} \frac{1}{\ell} - 1 \right)
%\end{align}
%which, for $N_r = N_t = n$, reduces to
%\begin{align}
%\Delta_{\mathcal{L}_\infty} = 3 \log_2 (e) \left( \sum_{\ell = 2}^n
%\frac{1}{\ell} \right)
%\end{align}
%This shows that for $N_r = N_t = n$ systems, the loss in spectral
%efficiency incurred by using ZF receivers can be particularly large,
%even for small $n$.  For instance, for $n = 2$ and $n = 3$, the loss
%is $2.16$ dB and $3.61$ dB respectively.

\subsubsection{Low SNR Analysis}

Recall that in the low SNR regime, the key channel-dependant parameter is the wideband slope.

\begin{proposition} \label{pr:IIDRayleigh_LowSNR}
For i.i.d.\ Rayleigh faded channels, the wideband slope for
MIMO MMSE receivers is given by
\begin{align} \label{eq:WidebandSlope_MMSE}
S_0^{\rm mmse} = \frac{2 N_r N_t}{2 N_t + N_r - 1} \; .
\end{align}
\end{proposition}
\vspace*{0.5cm}
\begin{proof}
For i.i.d.\ Rayleigh fading, using \cite[Lemma 6]{lozano03}, we find that
\begin{align}
\zeta \left( \mathbf{H} \mathbf{H}^\dagger \right) = \frac{N_r + N_t}{ N_t}  ,  \; \hspace*{1cm} \zeta ( \mathbf{H}_1 \mathbf{H}_1^\dagger ) = \frac{N_r + N_t - 1}{ N_t - 1}
%&E_{\mathbf{H}_i} \left[ {\rm tr}\left( [ \mathbf{H}_i \mathbf{H}_i^\dagger ]^2 \right) \right] = \frac{1}{ N_r (N_t-1)^2} N_r ( N_t - 1) (N_r + N_t %- 1)
\; .
\label{eq:IIDRay_Dispersion}
\end{align}
Substituting (\ref{eq:IIDRay_Dispersion}) into (\ref{eq:WideSlope_MMSE_IID}) leads to the result.
\end{proof}

This agrees with a recent result obtained via different methods in \cite{Louie08}.
%
%In this case, upon noting that \cite[Lemma 6]{Lozano03},
%\begin{align}
%E_{\mathbf{H}} \left[ {\rm tr}\left( [ \mathbf{H} \mathbf{H}^\dagger ]^2 \right) \right] = N_r N_t (N_r + N_t) , \; \;  E_{\mathbf{H}_i} \left[ {\rm %tr}\left( [ \mathbf{H}_i \mathbf{H}_i^\dagger ]^2 \right) \right] = N_r ( N_t - 1) (N_r + N_t - 1)
%\end{align}
%we obtain
%\begin{align} \label{eq:C0ddotMMSE}
%\ddot{C}^{\rm mmse}(0) = \frac{ - N_r ( 2 N_t + N_r - 1 )}{ N_t \, \ln 2} \; .
%\end{align}
It is interesting to compare (\ref{eq:WidebandSlope_MMSE}) with the corresponding wideband slope for optimal MIMO reception, given in \cite{lozano03} as
\begin{align} \label{eq:WidebandSlope_Optimal}
S_0^{\rm opt} = \frac{ 2 N_r N_t}{N_t + N_r} \; \; .
\end{align}
In Fig.\ \ref{fig:LowSNR_IID_Rayleigh},  the low SNR achievable rate approximations for MMSE and optimal receivers are presented, based on (\ref{eq:WidebandSlope_MMSE}) and (\ref{eq:WidebandSlope_Optimal}) respectively. The curves are shown as a function of received $\frac{E_b}{N_0}$, for a system with $N_r = N_t = 3$.  In both cases, the corresponding exact low SNR curves are also presented for further comparison, obtained by numerically solving (\ref{eq:LowSNRDefn}) and (\ref{eq:EbNoNonLinear}).  The figure shows that the linear approximations are accurate over a quite moderate range of $\frac{E_b^r}{N_0}$ values, especially for the MMSE receiver.

\begin{figure} \centering
\includegraphics[width=0.7\columnwidth]{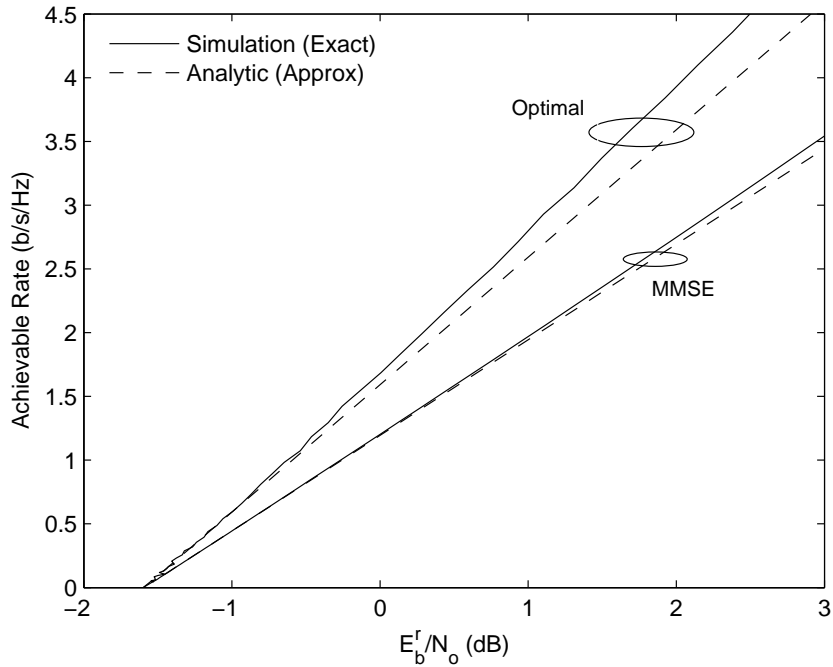}
\caption{Comparison of the spectral efficiency of a MIMO system with optimal and MMSE receivers in i.i.d.\ Rayleigh fading channels. Results are shown as a function of received $\frac{E_b}{N_0}$, for $N_r = N_t = 3$.}
 \label{fig:LowSNR_IID_Rayleigh}
\end{figure}

Clearly, from (\ref{eq:WidebandSlope_MMSE}), $S_0^{\rm mmse}$ is increasing in both $N_r$ and $N_t$, with the rate of increase being more significant for $N_r$.  This is in contrast to $S_0^{\rm opt}$, in which case both $N_r$ and $N_t$ play symmetric roles.  We also see that
%\begin{align}
%\frac{{E_b }}{{N_0 }}_{\min }^{\rm MIMO} = \frac{\ln 2}{N_r}, \; \hspace*{1cm} S_0^{\rm MIMO} = \frac{N_r N_t}{N_t + N_r} \; \; .
%\end{align}
%Thus, we see that the MIMO-MMSE receiver achieves the optimal (minimum possible) $\frac{{E_b }}{{N_0 }}_{\min }$.  The suboptimality of MMSE receivers at low SNR, however, is reflected in the wideband slope $S_0$ which satisfies
\begin{align}
\frac{ S_0^{\rm mmse} }{ S_0^{\rm opt} } = \frac{N_t + N_r}{ 2 N_t + N_r - 1}  \;  .
\end{align}
This ratio is increasing in $N_r$ and decreasing in $N_t$, satisfying
%which satisfies
\begin{align} \label{eq:ratio}
\frac{1}{2} \; \;  \leq \; \; \frac{N_t + N_r}{ 2 N_t + N_r - 1}  \; \; \leq \;  \; 1
\end{align}
where the lower bound is approached as $N_t \to \infty$ for fixed $N_r$, confirming that, relative to optimal receivers, MMSE receivers take a significant hit in the low SNR regime when the number of transmit antennas exceed the number of receive antennas. This is due primarily to the limited interference suppression capabilities of the receive array in this ``overloaded" scenario. On the other hand, the upper bound is achieved strictly for $N_t = 1$. It is also approached as ${N_r \to \infty}$ for fixed $N_t$,
%ratio (\ref{eq:ratio}) is monotonically decreasing in $N_r$, with
%\begin{align}
%\lim_{N_r \to \infty} \frac{ S_0^{\rm mmse} }{ S_0^{\rm opt} } = 1 \; ,
%\end{align}
revealing the intuitive notion that linear MIMO MMSE receivers perform near-optimally if the number of receive antennas are much larger than the number of transmit antennas, due, once again, to the additional captured received power and the enhanced interference suppression capabilities of the receive array.

Finally, it is interesting to consider the large-antenna regime.
%In this case we obtain the following remarkably simple expression:
\begin{corollary}
For the i.i.d.\ Rayleigh faded channel, as the number of antennas grows with ratio $\beta = \frac{N_t}{N_r}$, the ratio between the MMSE wideband slope (\ref{eq:WidebandSlope_MMSE}) and the optimal wideband slope (\ref{eq:WidebandSlope_Optimal}) converges to
\begin{align}
\frac{ S_0^{\rm mmse} }{ S_0^{\rm opt} } \to \frac{1 + \beta}{1 + 2 \beta} \; ,
\label{eq:ratio_S0}
\end{align}
which, interestingly, for $\beta = 1$ (i.e.\ $N_r = N_t$) gives
\begin{align}
\frac{ S_0^{\rm mmse} }{ S_0^{\rm opt} } \to \frac{2}{3} \; .
\label{eq:ratio_S0_simple}
\end{align}
\end{corollary}

\subsection{Correlated Rayleigh Fading}

We now particularize the general results of Section \ref{sec:GenResults} to spatially-correlated Rayleigh fading channels,  representative of non-line-of-sight environments with a lack of scattering around the transmitter and/or receiver, or with closely spaced antennas (with respect to the wavelength of the signal). We consider the popular ``separable" correlation model, described by
\begin{align} \label{eq:CorrRay_Model}
\mathbf{H} \sim \mathcal{CN}_{N_r, N_t} \left( \mathbf{0}, \mathbf{R} \otimes \mathbf{S} \right) \; ,
\end{align}
where $\mathbf{R}$ and $\mathbf{S}$ are Hermitian positive-definite matrices which represent, respectively, the receive and transmit spatial correlation.  This model, commonly adopted due to its analytic tractability, has also been confirmed through various measurement campaigns\footnote{Note that in some cases deviations from this model have also been observed \cite{Ozcelik03}.} (see, eg.\ \cite{Martin00,Kermoal02}).

It is important to note that, to our knowledge, the results in this section present the first analytical investigation of the achievable sum rate of MIMO MMSE receivers in the presence of spatial correlation.

\subsubsection{Exact Analysis}

For our exact analysis, we focus on \emph{semi-correlated} scenarios, allowing for spatial correlation at either the transmitter or receiver (but not both). We note, however, that the same approach can also be applied to derive closed-form solutions for the more general case in (\ref{eq:CorrRay_Model}), i.e.\ allowing for correlation at \emph{both} the transmitter and receiver, by employing the MIMO ergodic mutual information results for such channels established in \cite{kiessling04,simon_04_submit}.  The final expressions, however, involve more cumbersome notation compared with the semi-correlated results, and as such we choose to omit them here. (Note that the more general model (\ref{eq:CorrRay_Model}) will be explicitly considered in the following subsections, when focusing on asymptotic SNR regimes.)
Throughout this subsection, we will denote the spatial correlation matrix, either receive or transmit, by the generic symbol $\mathbf{L}$.

Before addressing the achievable sum rate of MMSE receivers, it is convenient to give the following new result for the ergodic mutual information with optimal MIMO receivers, which simplifies and unifies prior expressions in the literature for semi-correlated Rayleigh fading.

\begin{lemma}  \label{le:CorrRayleigh}
Consider the transmit-correlated Rayleigh channel $\mathbf{H} \sim \mathcal{CN}_{N_r, N_t} ( \mathbf{0}, \mathbf{I}_{N_r} \otimes \mathbf{L})$, or receive-correlated Rayleigh channel $\mathbf{H} \sim \mathcal{CN}_{N_r, N_t} ( \mathbf{0}, \mathbf{L} \otimes \mathbf{I}_{N_t} )$, where the spatial correlation matrix $\mathbf{L}$ has dimension $q \times q$ (i.e.\ for transmit-correlation $q = N_t$, for receive-correlation $q = N_r$), with eigenvalues $\beta_1 > \cdots > \beta_q$. Also, let $p \in \{ N_r, N_t\} \backslash q$.
Then the ergodic MIMO mutual information of $\mathbf{H}$ with isotropic inputs and optimal receivers is given by
\begin{align}
I^{\rm opt} (N_r, N_t, {\rm snr}) = \frac{ \log_2e }{ \prod_{\ell < k}^q ( \beta_k - \beta_\ell ) }  \sum_{k = q - n + 1}^q \det \mathbf{E}_{p,q}(k)
\end{align}
where $\mathbf{E}_{p,q}(k)$ is a $q \times q$ matrix with $(s,t)$th entry
\begin{align} \label{eq:Ek_defn}
\left( \mathbf{E}_{p,q}(k) \right)_{s,t} = \left\{
\begin{array}{ll}
\beta_{s}^{t-1} &, \, t \neq k \\
\beta_s^{t-1} e^{\frac{N_t}{ \beta_s {\rm snr}}} \sum_{h = 1}^{p-q+t} {\rm E}_h \left( \frac{N_t}{\beta_s {\rm snr} } \right)  &, \,   t = k \\
\end{array}
\right. \; \; .
\end{align}
\end{lemma}
\begin{proof}
See Appendix \ref{app:CorrRayleigh}.
\end{proof}

It is important to note that \emph{Lemma \ref{le:CorrRayleigh}} allows the correlation to occur between the transmit or receive antennas, and places no restrictions on the system dimensions.  This is in contrast to prior analyzes (see \cite{smith03,chiani03,alfano04}) which have given a separate treatment depending on whether the correlation occurs at the end of the link with the least or most number of antennas.

Given \emph{Lemma \ref{le:CorrRayleigh}}, we can now obtain exact closed-form solutions for the achievable sum rate of MIMO MMSE receivers in semi-correlated Rayleigh fading.  It is convenient to treat the cases of transmit and receive correlation separately.

\begin{proposition} \label{pr:RxCorr}
Let $\mathbf{H} \sim \mathcal{CN}_{N_r, N_t} ( \mathbf{0}, \mathbf{L} \otimes \mathbf{I}_{N_t} )$, with $\mathbf{L}$ defined as above. %Also, define $s = \min(N_r, N_t)$ and $\tilde{s} = \min(N_r, N_t-1)$.
Then the MMSE achievable sum rate is given by
\begin{align}
I^{\rm mmse}(N_r, N_t, {\rm snr}) = \frac{ N_t \log_2e }{ \prod_{\ell < k}^{N_r} ( \beta_k - \beta_\ell ) } \left(  \sum_{k = N_r - n + 1}^{N_r} \det \mathbf{E}_{N_t,N_r}(k) - \sum_{k = N_r - n' + 1}^{N_r} \det \mathbf{E}_{N_t-1,N_r}(k) \right) \; .
\label{eq:RxCorrelation}
\end{align}
%\begin{align}
%I^{\rm mmse}(N_r, N_t, {\rm snr}) = \frac{ N_t \log_2e }{ \prod_{\ell < k}^{N_r} ( \beta_k - \beta_\ell ) } \left(  \sum_{k = N_r - n + 1}^{N_r} \det %\mathbf{E}_{N_t,N_r}(k) - \sum_{k = N_r - n' + 1}^{N_r} \det \mathbf{E}_{N_t-1,N_r}(k) \right) \; .
%\end{align}
\end{proposition}
\begin{proof}
The first term in (\ref{eq:MMSESumCapacity_Main}) is evaluated directly from \emph{Lemma \ref{le:CorrRayleigh}}.  The remaining terms are directly inferred from \emph{Lemma \ref{le:CorrRayleigh}}, upon noting that $\mathbf{H}_i \sim \mathcal{CN}_{N_r, N_t-1}(\mathbf{0}, \mathbf{L} \otimes \mathbf{I}_{N_t-1})$.
\end{proof}

\begin{proposition} \label{pr:TxCorr}
Let $\mathbf{H} \sim \mathcal{CN}_{N_r, N_t} ( \mathbf{0}, \mathbf{I}_{N_r} \otimes \mathbf{L})$, with $\mathbf{L}$ defined as above.
Also, let $\mathbf{L}^{ii}$ denote the $(i,i)$th minor of $\mathbf{L}$, with eigenvalues $\beta_{i,1} > \cdots > \beta_{i,N_t-1}$.  Then the MMSE achievable sum rate is given by
\begin{align}
I^{\rm mmse}(N_r, N_t, {\rm snr}) &= \frac{ N_t \log_2e }{ \prod_{\ell < k}^{N_t} ( \beta_k - \beta_\ell ) }  \sum_{k = N_t - n + 1}^{N_t} \det \mathbf{E}_{N_r,N_t}(k) \nonumber \\
 & \hspace*{2cm} - \sum_{i=1}^{N_t} \frac{ \log_2e }{ \prod_{\ell < k}^{N_t-1} ( \beta_{i,k} - \beta_{i,\ell} ) }  \sum_{k = N_t - n'}^{N_t-1} \det \mathbf{E}_{N_r, N_t-1}(k,i) \; ,
\label{eq:TxCorrelation}
\end{align}
where $\mathbf{E}_{N_r, N_t-1}(k,i)$ is defined as in (\ref{eq:Ek_defn}), but with $\beta_{i,k}$ replacing $\beta_{k}$.
\end{proposition}
\begin{proof}
The first term in (\ref{eq:MMSESumCapacity_Main}) is evaluated directly from \emph{Lemma \ref{le:CorrRayleigh}}. The remaining terms are directly inferred from \emph{Lemma \ref{le:CorrRayleigh}}, upon noting that $\mathbf{H}_i \sim \mathcal{CN}_{N_r, N_t-1}(\mathbf{0}, \mathbf{I}_{N_r} \otimes \mathbf{L}^{ii})$.
\end{proof}
It is important to note that the results in \emph{Propositions \ref{pr:RxCorr}} and \emph{\ref{pr:TxCorr}} apply for arbitrary numbers of transmit and receive antennas.

The result in \emph{Proposition \ref{pr:TxCorr}} is demonstrated in Fig.\ \ref{fig:MMSE_SumCapacity_Corr_Rayleigh}, where it is compared with the exact MMSE achievable sum rate based on Monte-Carlo simulations, for two different transmit-correlation scenarios.  Here, the simple exponential correlation model was employed, in which case the correlation matrix $\mathbf{L}$ was constructed with $(i,j)$th entry $\rho^{|i-j|}$, with $\rho$ denoting the correlation coefficient.  We see a precise agreement with the analysis and simulations, as expected.
Moreover, the MMSE achievable sum rate is seen to degrade as the level of transmit correlation is increased, especially in the high SNR regime.  %This phenomenon will now be explored in more detail.

\begin{figure} \centering
\includegraphics[width=0.7\columnwidth]{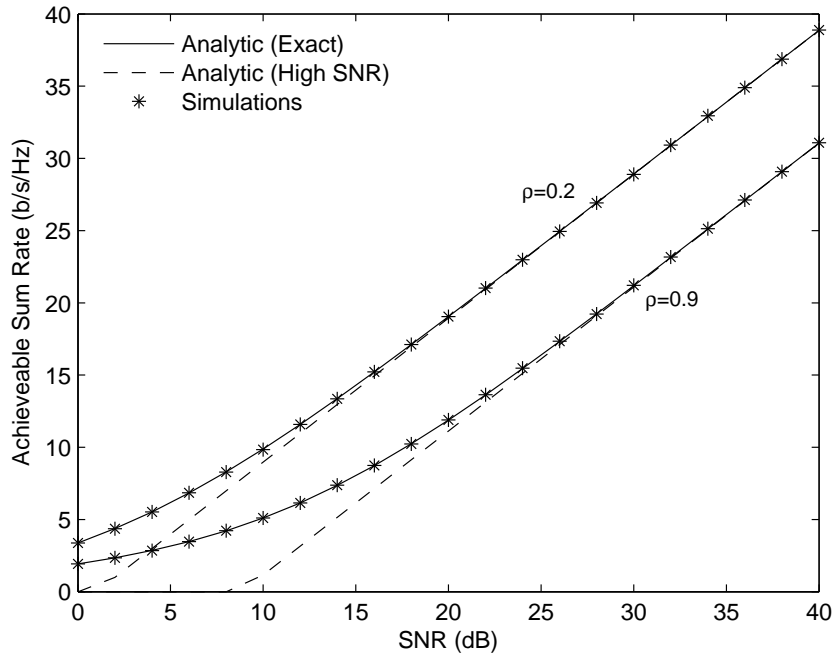}
\caption{Achievable sum rate of MIMO MMSE receivers in transmit-correlated Rayleigh fading; comparison of analysis and simulations.  Results are shown for $N_t = 3$ and $N_r = 5$, and for different correlation coefficients $\rho$.}
 \label{fig:MMSE_SumCapacity_Corr_Rayleigh}
\end{figure}

%\emph{Comments on Alternative MIMO Fading Models:}
%
%The same approach can be applied to derive closed-form solutions for many other fading models of interest.
%For example, for Rayleigh fading channels with \emph{both} transmit and receive correlation and uncorrelated Rician %fading channels, exact expressions for the MMSE sum capacity can be directly evaluated using MIMO capacity results %from
%\cite{kiessling04,simon_04_submit} and \cite{alfano04_2,simon_04_submit,kang04} respectively.  In both cases, the %final expressions involve introducing more cumbersome notation (although they are in terms of standard functions which %are easy to compute), so we choose to omit these here.

\subsubsection{High SNR Analysis}

Here we consider channels of the general form (\ref{eq:CorrRay_Model}).  The key focus, once again, is on the high SNR power offset for the case $N_r \geq N_t$.  This is given by the following key result:
%\begin{align} \label{eq:DoubleCorrelated}
%\mathbf{H} = \mathbf{R}^{1/2} \mathbf{Z} \mathbf{S}^{1/2}
%\end{align}
%where $\mathbf{Z} \sim \mathcal{C N}_{N_r, N_t} ( \mathbf{0}, \mathbf{I}_{N_r} \otimes \mathbf{I}_{N_t} )$.  Also, we consider the case $N_r \geq N_t$, and assume that $\mathbf{R}$ and $\mathbf{S}$ are Hermitian positive-definite matrices with distinct eigenvalues $ \mathbf{r} = \{ r_1, \cdots, r_{N_r} \}$ and $\mathbf{s} = \{ s_1, \cdots, s_{N_t} \}$ respectively.  We have the following key result:
\begin{proposition} \label{th:highSNRPO_CorrRay}
For transmit and receive correlated Rayleigh faded channels, the high SNR power offset for
a MIMO system with MMSE receiver is given by
%\begin{align}
%\mathcal{L}_\infty^{\rm mmse} &= \log_2 (N_t) - \log_2 (e) \left(
%\sum_{\ell=1}^{N_t - 1} \frac{1}{\ell} - \gamma \right) + f(\mathbf{S}) - g(\mathbf{R})
%\label{eq:L_MMSE_final}
%\end{align}
\begin{align}
\mathcal{L}_\infty^{\rm mmse}(\mathbf{R}, \mathbf{S}) &= \mathcal{L}_\infty^{\rm mmse}(\mathbf{I}_{N_r}, \mathbf{I}_{N_t}) + f(\mathbf{S}) + g(\mathbf{R})
\label{eq:L_MMSE_final}
\end{align}
where $\mathcal{L}_\infty^{\rm mmse}(\mathbf{I}_{N_r}, \mathbf{I}_{N_t})$ is the power offset in the absence of spatial correlation given in (\ref{eq:L_ZF_final}), and $f(\cdot)$ and $g(\cdot)$ are given by
\begin{align}
f(\mathbf{S}) = \frac{ 1 }{N_t} \sum_{k=1}^{N_t} \log_2 [ \mathbf{S}^{-1}]_{k,k}
\label{eq:Seffect}
\end{align}
and % $[\cdot]_{k,k}$ denoting the $k$th diagonal element, and
%\begin{align}
%g(\mathbf{R}) = \frac{ \det ( \mathbf{Y}_{N_r-N_t+1} (\mathbf{r}) )}{ \prod_{i < j}^{N_r} ( r_j - r_i ) }
%\label{eq:Reffect}
%\end{align}
\begin{align}
g(\mathbf{R}) = \log_2e \left( \sum_{\ell=1}^{N_r - N_t} \frac{1}{\ell} - \sum_{\ell=1}^{N_t-1} \frac{1}{\ell} \right) - \frac{ \det \mathbf{Y}_{N_r-N_t+1} (\mathbf{r}) }{ \prod_{i < j}^{N_r} ( r_j - r_i ) }
\label{eq:Reffect}
\end{align}
respectively.  Here, $\mathbf{r} = (r_1, \ldots, r_{N_r})^T$, with $r_1 > \ldots > r_{N_r}$, are the eigenvalues of $\mathbf{R}$, and $\mathbf{Y}_{N_r-N_t+1}(\mathbf{r})$ denotes an $N_r \times N_r$ matrix with $(s, t)$th element
\begin{align} \label{eq:Yr_Defn}
( \mathbf{Y}_{N_r-N_t+1}(\mathbf{r}) )_{s,t} = \left\{
\begin{array}{lr}
r_s^{t-1} & {\rm for} \; \, t \neq N_r-N_t+1 \\
r_s^{t-1} \log_2 r_s & {\rm for} \; \, t = N_r-N_t+1
\end{array}
\right.  .
\end{align}
\end{proposition}
\vspace*{0.5cm}
\begin{proof}
The result is easily obtained from (\ref{eq:MMSE_HighSNR}), upon invoking the following result\footnote{Note that an equivalent expression for (\ref{eq:expNew}) can be found in \cite{lozano05_jnl}.} \cite{jin07_IT}
\begin{align}
E_{\mathbf{H}} \left[ \mathcal{J} ( N_r, N_t, \mathbf{H} ) \right] &=  \log_2 \det \mathbf{S} + \log_2e \sum_{\ell=1}^{N_t} \psi (\ell) +  \frac{ \sum_{j=N_r-N_t+1}^{N_r} \det  \mathbf{Y}_j(\mathbf{r}) }{ \prod_{i<j}^{N_r} ( r_j - r_i ) } \;
\label{eq:expNew}
\end{align}
and noting that $\mathbf{H}_k \sim \mathcal{C N}_{N_r, N_t-1}(\mathbf{0}, \mathbf{R} \otimes \mathbf{S}^{kk})$, where $\mathbf{S}^{kk}$ is the $(k,k)$th minor of $\mathbf{S}$.
\end{proof}

%\begin{proof}
%See Appendix \ref{Ap:highSNRPO_CorrRay}.
%\end{proof}

\begin{figure} \centering
\includegraphics[width=0.7\columnwidth]{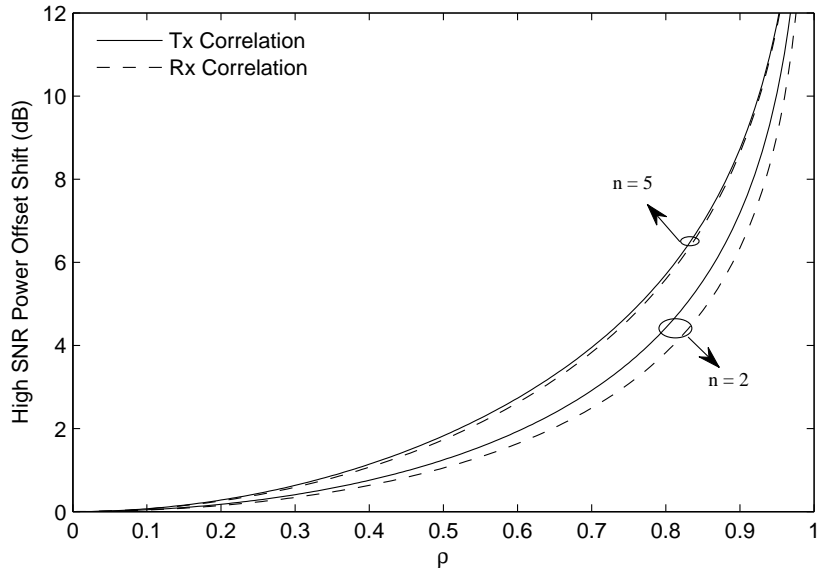}
\caption{Shift in the high SNR power offset of MIMO MMSE receivers due to spatial correlation. Results are shown for both transmit and receive correlation, as a function of correlation coefficient $\rho$, for different $N_r = N_t = n$.}
 \label{fig:HighSNR_Offset_Corr}
\end{figure}

Thus, in the high SNR regime, the effect of both transmit and receive correlation are clearly decoupled, being determined by the functions $f(\cdot)$ and $g(\cdot)$ respectively. Considering the case of transmit correlation, we see that
\begin{align}
f(\mathbf{S}) \geq \frac{ 1 }{N_t} \log_2 {\rm det}^{-1} ( \mathbf{S} ) \geq 0
\label{eq:fSBound}
\end{align}
with equality for $\mathbf{S} = \mathbf{I}_{N_t}$.  This result can be obtained from (\ref{eq:Seffect}) upon noting that (i) the set of diagonal elements $[ \mathbf{S}^{-1}]_{k,k}$ are majorized by the set of eigenvalues of $ \mathbf{S}^{-1} $, which, in turn are the reciprocals of the eigenvalues of $\mathbf{S}$, and (ii) the function $\sum_i \log_2 x_i$ is Schur-concave.  Thus, we see from (\ref{eq:fSBound}) that transmit correlation, whilst not affecting the high SNR slope (\ref{eq:MMSEHighSNRSlope}), reduces the achievable sum rate of MMSE receivers in the high SNR regime through an increased high SNR power offset; as already observed experimentally in Fig.\ \ref{fig:MMSE_SumCapacity_Corr_Rayleigh}.  The high SNR approximation based on (\ref{eq:HighSNRGeneral}) and (\ref{eq:L_MMSE_final}) is also presented in Fig.\ \ref{fig:MMSE_SumCapacity_Corr_Rayleigh}, and is seen to converge to the exact MMSE achievable sum rate for quite moderate SNR levels.
%This indicates that the capacity loss due to correlation predicted by the power offset (\ref{eq:L_MMSE_final}), although formally valid for asymptotically large SNRs, gives a good concise representation of the effect of correlation in various practical systems.

In contrast to the case of transmit correlation, the overall impact of receive correlation is not immediately evident from $g(\cdot)$ in (\ref{eq:Reffect}), due mainly to the presence of the Vandermonde determinant in the denominator of the second term.  This expression does reveal, however, that the effect of $\mathbf{R}$ is purely through its eigenvalues, and, interestingly, the relative impact of receive correlation not only depends on $N_r$, but also on $N_t$.  This is in contrast to the effect of transmit correlation in $f(\cdot)$, which depends only on $N_t$.

Fig.\ \ref{fig:HighSNR_Offset_Corr} plots the shift in high SNR power offset (in dB) due to transmit correlation, based on (\ref{eq:Seffect}), and the shift due to receive correlation, based on (\ref{eq:Reffect}), as a function of the correlation coefficient $\rho$, where $\mathbf{S}$ and $\mathbf{R}$ are constructed according to the exponential correlation model with $(i,j)$th elements $\mathbf{S}_{i,j} = \rho^{|i-j|}$ and $\mathbf{R}_{i,j} = \rho^{|i-j|}$ respectively.  From the figure, we can conclude that the SNR penalty for a MIMO MMSE system increases with the level of transmit or receive correlation.  Interestingly, the figure also shows that for a given correlation coefficient $\rho$, the SNR penalty (for the $N_r = N_t = n$ scenarios considered) is more severe if the correlation occurs at the transmitter, rather than the receiver; with this difference being most significant for small $n$.

Now consider the excess high SNR power offset.  To evaluate this, we require the high SNR power offset with optimal receivers $\mathcal{L}_\infty^{\rm opt}$, which for the transmit-receive correlated case was first presented in \cite[Eq. 28]{lozano05_jnl}.  Using a result from \cite{jin07_IT}, an alternative simplified expression can be obtained, as given by the following lemma.
\begin{lemma}
For transmit and receive correlated Rayleigh faded MIMO channels, the high SNR power offset with optimal receivers is given by
\begin{align} \label{eq:opt_LInf}
\mathcal{L}_\infty^{\rm opt} &= \log_2N_t - \log_2e \left( \sum_{\ell=2}^{N_t} \frac{1}{\ell} - \gamma \right) -  \frac{1}{N_t} \left( \log_2 \det \mathbf{S}  + \frac{ \sum_{j=N_r-N_t+1}^{N_r} \det  \mathbf{Y}_j (\mathbf{r})}{ \prod_{i<j} ( r_j - r_i ) }  \right) \;
\end{align}
which, for the special case $N_r = N_t = n$, reduces to\footnote{This special case was also reported in \cite[Eq. 28]{lozano05_jnl}.}
\begin{align}
\mathcal{L}_\infty^{\rm opt} &= \log_2n - \log_2e \left( \sum_{\ell=2}^{n} \frac{1}{\ell} - \gamma \right) -  \frac{1}{n} \left( \log_2 \det \mathbf{S}   + \log_2 \det \mathbf{R}  \right) \; .
\end{align}
\end{lemma}

\vspace*{0.5cm}

The excess high SNR power offset is now readily obtained from (\ref{eq:opt_LInf}) and (\ref{eq:L_MMSE_final}).
\begin{corollary}
For transmit and receive correlated Rayleigh faded MIMO channels, the excess high SNR power
offset is given by
\begin{align}
\Delta_{\rm ex} &= - \log_2e \frac{N_t-1}{N_t} + g_1(\mathbf{S}) + g_2(\mathbf{R})
\end{align}
where
\begin{align} \label{eq:g1Excess}
g_1(\mathbf{S}) = \frac{1}{N_t} \left( \sum_{k=1}^{N_t} \log_2 \left[ \mathbf{S}^{-1} \right]_{k,k} + \log_2 \det \mathbf{S} \right)
\end{align}
and
\begin{align}
g_2 & (\mathbf{R}) = \frac{  \sum_{j=N_r-N_t+2}^{N_r} \det \mathbf{Y}_j(\mathbf{r}) - (N_t-1) \det \mathbf{Y}_{N_r-N_t+1}(\mathbf{r})  }{N_t \prod_{i<j} ( r_j - r_i )}  \; .
\end{align}

\end{corollary}

\vspace*{0.5cm}

From (\ref{eq:fSBound}) and (\ref{eq:Seffect}), it is easy to establish that (\ref{eq:g1Excess}) is non-negative, ie.\ $g_1(\mathbf{S}) \geq 0$, indicating that in the high SNR regime MMSE receivers incur a more significant rate loss due to transmit correlation, compared with optimal MIMO receivers.

\subsubsection{Low SNR Analysis}

For our low SNR analysis, we consider channels of the general form (\ref{eq:CorrRay_Model}).  In this regime, the main focus, once again, is on characterizing the wideband slope.
\begin{proposition} \label{pr:CorrRayleigh_LowSNR}
For transmit and receive correlated Rayleigh faded channels, the wideband slope for
MIMO MMSE receivers is given by
\begin{align} \label{eq:S0CorrClosedForm}
S_0^{\rm mmse} = \frac{ 2 N_r N_t }{ (2 N_t - 1) \zeta(\mathbf{R}) + N_r \left( N_t \zeta (\mathbf{S}) - \frac{(N_t-1)}{N_t} \sum_{i=1}^{N_t} \zeta (\mathbf{S}^{ii}) \right) } \;
\end{align}
where $\mathbf{S}^{ii}$ is the $(i,i)$th minor of $\mathbf{S}$.
\end{proposition}

\begin{proof}
For the correlated Rayleigh fading model (\ref{eq:CorrRay_Model}), we can infer the following from \cite{lozano03},
\begin{align}
 & \zeta (\mathbf{H} \mathbf{H}^\dagger ) =   \zeta ( \mathbf{R} ) + \frac{N_r}{N_t} \zeta ( \mathbf{S} ) \; , \nonumber \\
 & \zeta (\mathbf{H}_i \mathbf{H}_i^\dagger ) = \zeta ( \mathbf{R} ) + \frac{N_r}{N_t-1} \zeta ( \mathbf{S}^{ii} ) \; .
\label{eq:corrRes}
\end{align}
Substituting (\ref{eq:corrRes}) into (\ref{eq:WideSlope_MMSE}) leads to the result.
\end{proof}

Note that since the diagonal elements of both $\mathbf{R}$ and $\mathbf{S}$ are unity, it follows from (\ref{eq:dispersionDefn}) that the dispersion numbers in (\ref{eq:corrRes}) particularize to
\begin{align}
\zeta(\mathbf{R}) = \frac{{\rm tr}(\mathbf{R}^2)}{N_r}, \; \; \; \; \zeta(\mathbf{S}) = \frac{{\rm tr}(\mathbf{S}^2)}{N_t},  \; \; \; \; \zeta(\mathbf{S}^{ii}) = \frac{{\rm tr}( (\mathbf{S}^{ii})^2)}{N_t-1} \; .
\end{align}

For the case of receive correlation only (i.e.\ $\mathbf{S} = \mathbf{I}_{N_t}$), (\ref{eq:S0CorrClosedForm}) admits the very simple form
\begin{align} \label{eq:MMSEResult}
S_0^{\rm mmse} = \frac{ 2 N_r N_t }{ (2 N_t - 1) \zeta(\mathbf{R}) + N_r  } \; .
\end{align}
Since $\zeta(\mathbf{R})$ satisfies
\begin{align} \label{eq:ZetaBound}
1 \leq \zeta(\mathbf{R}) \leq N_r
\end{align}
with the lower bound achieved if the antennas are uncorrelated and the upper bound achieved if the antennas are fully correlated, we see from (\ref{eq:MMSEResult}) that receive correlation reduces the achievable sum rate of MMSE receivers in the low SNR regime, as quantified by a reduction in wideband slope.  It is also interesting to compare (\ref{eq:MMSEResult}) with the wideband slope for MIMO with optimal receivers, given by \cite{lozano03}
\begin{align}
S_0^{\rm opt} = \frac{ 2 N_r N_t }{ N_t \zeta(\mathbf{R}) + N_r } \; \;  \; .
\end{align}
Thus, we have the ratio
\begin{align} \label{eq:RatioCorr}
\frac{ S_0^{\rm mmse}}{S_0^{\rm opt}} = \frac{ N_t \zeta(\mathbf{R}) + N_r }{ (2 N_t - 1) \zeta(\mathbf{R}) + N_r  } \;
\end{align}
which, based on (\ref{eq:ZetaBound}), also decreases with receive correlation, satisfying
\begin{align}
 \frac{N_t + 1}{2 N_t} \; \;   \leq  \; \;  \frac{ S_0^{\rm mmse}}{S_0^{\rm opt}}  \; \; \leq   \; \;  \frac{ N_t + N_r }{ 2 N_t + N_r - 1 } \; .
\end{align}
This result indicates that not only are  MMSE receivers degraded at low SNR due to receive correlation, but they actually incur \emph{more} of a loss than do optimal MIMO receivers. Interestingly, we also see that
\begin{align}
\lim_{N_r \to \infty} \frac{ S_0^{\rm mmse}}{S_0^{\rm opt}} = 1 \;
\end{align}
and, for $\zeta(\mathbf{R}) \neq 0$,
\begin{align}
\lim_{N_t \to \infty} \frac{ S_0^{\rm mmse}}{S_0^{\rm opt}} = \frac{1}{2} \; ,
\end{align}
which is the same limiting behavior observed previously for uncorrelated Rayleigh channels in (\ref{eq:ratio}).

For the case of transmit correlation only (i.e.\ $\mathbf{R} = \mathbf{I}_{N_r}$), focusing on the scenario $N_t = 2$, (\ref{eq:S0CorrClosedForm}) reduces to
\begin{align} \label{eq:MMSEResult_TxCorr}
S_0^{\rm mmse} = \frac{ 4 N_r }{ 3 + N_r ( 2 \zeta(\mathbf{S}) - 1 )  } \; \; .
\end{align}
This result reveals that in the low SNR regime, the effect of transmit correlation in the channel mirrors that of receive correlation by reducing the MMSE achievable sum rate through a reduction in wideband slope.

\subsection{Uncorrelated Rician Fading}

We now particularize the general results of Section \ref{sec:GenResults} to Rician fading channels, representative of line-of-sight environments.  For convenience, we focus on uncorrelated Rician channels with rank-$1$ specular component, described by
\begin{align} \label{eq:RicianModel}
%\mathbf{H} = \sqrt{\frac{K}{K+1}} \mathbf{a}(\theta_r) \mathbf{a}^T(\theta_t) + \sqrt{\frac{1}{K+1}} \mathbf{Z}
\mathbf{H} \sim \mathcal{CN}_{N_r, N_t} \left(  \sqrt{\frac{K}{K+1}} \mathbf{a}(\theta_r) \mathbf{a}^T(\theta_t),  \frac{1}{K+1} \mathbf{I}_{N_r} \otimes \mathbf{I}_{N_t} \right)
\end{align}
%where $\mathbf{Z} \sim \mathcal{C N}_{N_r, N_t} ( \mathbf{0}, \mathbf{I}_{N_r} \otimes \mathbf{I}_{N_t} )$,
where $K$ is the Rician $K$-factor, and $\mathbf{a}(\cdot)$ denotes an array response (column) vector (see \cite[Eq. 5]{bolcskei03}), parameterized by the angle of arrival $\theta_r$ and angle of departure $\theta_t$ of the specular component.

It is important to note that, to our knowledge, the results in this section present the first analytical investigation of the achievable sum rate of MIMO MMSE receivers in the presence of Rician fading.

\subsubsection{Exact Analysis}

An exact expression for the MMSE achievable sum rate can be easily obtained by evaluating (\ref{eq:MMSESumCapacity_Main}) using exact results for the ergodic mutual information of MIMO Rician channels with optimal receivers, given in \cite{alfano04_2,simon_04_submit,kang04}. We choose to omit explicitly presenting this result here to avoid the introduction of more cumbersome notation.

% respectively.  In both cases, the %final expressions involve introducing more cumbersome notation (although they are in terms of standard functions which %are easy to compute), so we choose to omit these here.

\subsubsection{High SNR Analysis}

In Rician fading, the high SNR power offset (for the case $N_r \geq N_t$) is given by the following key result:
\begin{proposition} \label{th:HighSNRPO_Rician}
For uncorrelated Rician faded channels, the high SNR power offset for
a MIMO system with MMSE receiver is given by
\begin{align}
\mathcal{L}_\infty^{\rm mmse}(K) &= \mathcal{L}_\infty^{\rm mmse}(0)
+ h_1(K)
\label{eq:L_MMSE_Rician}
\end{align}
where $\mathcal{L}_\infty^{\rm mmse}(0)$ is given by (\ref{eq:L_ZF_final}) and
\begin{align}
h_1(K) &= \log_2 (K+1) - K \log_2e \bigl( N_t \theta(N_r, N_t, K) - (N_t-1) \theta(N_r, N_t-1, K) \bigr)  \, ,
\end{align}
with
\begin{align}
\theta(N_r, N_t, K) =  {}_2F_2(1,1;2,N_r+1; -K N_r N_t)
\end{align}
where ${}_2F_2(\cdot)$ denotes the generalized hypergeometric function.
\end{proposition}
\begin{proof}
The result is easily obtained from (\ref{eq:MMSE_HighSNR}), upon invoking the following result \cite{lozano05_jnl}
\begin{align}
E_{\mathbf{H}}  \left[ \mathcal{J}(N_r, N_t, \mathbf{H} \right] &=  \log_2e \sum_{\ell=0}^{N_t-1} \psi(N_r - \ell) - N_t \log_2(K+1)  + K N_t \log_2e \theta(N_r, N_t, K)
\end{align}
and noting that
\begin{align}
\mathbf{H}_k \sim \mathcal{CN} \left( \sqrt{\frac{K}{K+1}} \mathbf{a}(\theta_r) \mathbf{a}_k^T(\theta_t),  \frac{1}{K+1} \mathbf{I}_{N_r} \otimes \mathbf{I}_{N_t-1} \right) \; ,
\end{align}
where $\mathbf{a}_k(\cdot)$ corresponds to the response vector $\mathbf{a}(\cdot)$ with the $k$th element removed.
\end{proof}

\begin{figure} \centering
\includegraphics[width=0.7\columnwidth]{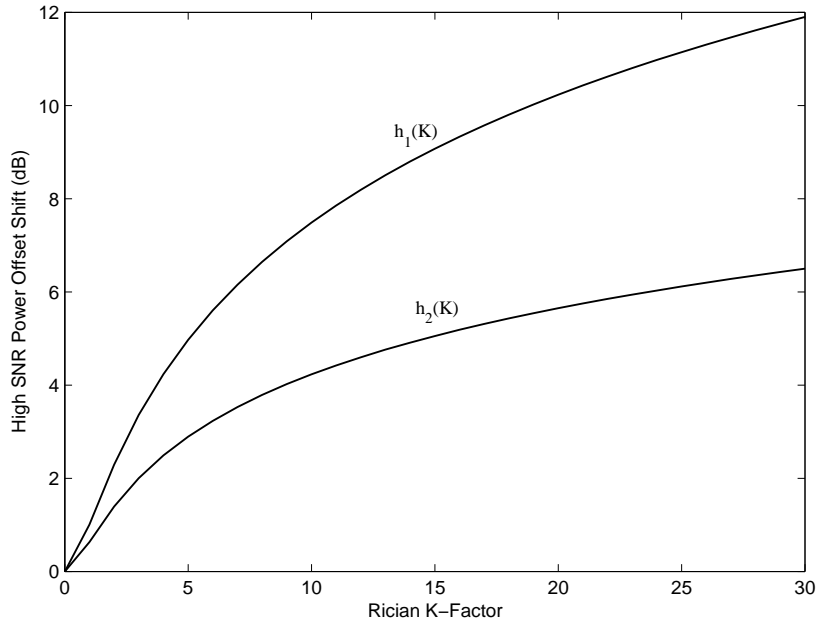}
\caption{Shift in high SNR power offset $h_1(K)$ and excess power offset $h_2(K)$ in Rician fading.  Results are shown for $N_r = N_t = 2$.}
 \label{fig:MMSE_HighSNR_Rician}
\end{figure}

Interestingly, we see that the impact of line-of-sight in the high SNR regime is only through the relative \emph{strength} of the specular component (ie.\ through the Rician $K$-factor), and is independent of its geometry (ie.\ independent of $\theta_r$ and $\theta_t$).  Moreover, \emph{Proposition \ref{th:HighSNRPO_Rician}} reveals that line-of-sight imposes a shift in power offset, as succinctly characterized through the function $h_1(K)$.  This behavior is investigated in Fig.\ \ref{fig:MMSE_HighSNR_Rician}, where we plot $h_1(K)$ (in dB) as a function of $K$.  We clearly see that the high SNR power offset increases monotonically with $K$, revealing that the presence of line-of-sight reduces the achievable sum rate of MIMO MMSE receivers in the high SNR regime.

\begin{corollary}
For uncorrelated Rician faded MIMO channels, the excess high SNR power
offset is given by
\begin{align}
\Delta_{\rm ex}(K) &= \Delta_{\rm ex}(0) + h_2(K)
\end{align}
with $\Delta_{\rm ex}(0)$ given by (\ref{eq:IIDRayExcessPO}), and
\begin{align}
h_2(K) &= - \log_2e K (N_t-1) \bigl( \theta(N_r, N_t, K) - \theta(N_r, N_t-1, K) \bigr) \; .
\end{align}
\end{corollary}
\begin{proof}
This result is obtained by substituting (\ref{eq:L_MMSE_Rician}) and \cite[Eq. (67)]{lozano05_jnl} into (\ref{eq:DeltaDefn}).
\end{proof}

As shown in Fig.\ \ref{fig:MMSE_HighSNR_Rician}, the function $h_2(K)$ increases with $K$, confirming that not only does line-of-sight decrease the high SNR achievable sum rate of MMSE receivers through an increased power offset, but the loss is more significant compared with optimal MIMO receivers.

\subsubsection{Low SNR Analysis}

For Rician channels in the low SNR regime, we have the following key result:
\begin{proposition} \label{pr:Rician_LowSNR}
For uncorrelated Rician faded channels, the wideband slope for
MIMO MMSE receivers is given by
\begin{align} \label{eq:S0RicClosedForm}
S_0^{\rm mmse} = \frac{ 2 N_r N_t (K + 1)^2 }{ K^2 ( 2  N_t - 1)N_r + (2K+1) ( 2 N_t + N_r - 1 ) } \; .
\end{align}
\end{proposition}
\begin{proof}
For the uncorrelated Rician fading model (\ref{eq:RicianModel}), we can infer the following from \cite{lozano03},
%\begin{align}
%& \zeta ( \mathbf{H} \mathbf{H}^\dagger  ) = \frac{ 1 }{(K+1)^2} \left( N_r K^2 + \frac{(N_r + N_t)(2 K + 1)}{N_t} \right)   \nonumber \\
%& \zeta \left( \mathbf{H}_i \mathbf{H}_i^\dagger \right) = \frac{ 1 }{(K+1)^2} \left( N_r K^2 + \frac{(N_r + N_t-1)(2 K + 1)}{N_t-1} \right) \; .
%\label{eq:RicTr}
%\end{align}
\begin{align}
& \zeta ( \mathbf{H} \mathbf{H}^\dagger  ) = \frac{ N_r K^2 + \frac{(N_r + N_t)(2 K + 1)}{N_t} }{(K+1)^2}    \nonumber \\
& \zeta \left( \mathbf{H}_i \mathbf{H}_i^\dagger \right) = \frac{ N_r K^2 + \frac{(N_r + N_t-1)(2 K + 1)}{N_t-1} }{(K+1)^2}  \; .
\label{eq:RicTr}
\end{align}
Substituting (\ref{eq:RicTr}) into (\ref{eq:WideSlope_MMSE}) leads to the result.
%
%\begin{align}
%&E_{\mathbf{H}} \left[ {\rm tr}\left( [ \mathbf{H} \mathbf{H}^\dagger ]^2 \right) \right] = \frac{ N_r N_t}{(K+1)^2} %(N_r N_t K^2 + (N_r + N_t)(2 K + 1) )   \nonumber \\
%&E_{\mathbf{H}_i} \left[ {\rm tr}\left( [ \mathbf{H}_i \mathbf{H}_i^\dagger ]^2 \right) \right] = \frac{ N_r %(N_t-1)}{(K+1)^2} (N_r (N_t-1) K^2 + (N_r + N_t-1)(2 K + 1) ) \; .
%\label{eq:RicTr}
%\end{align}
\end{proof}
%
%We now investigate the MMSE sum capacity at low SNR in Rician fading channels of the form (\ref{eq:RicianModel}).
%In this case, we can again use results from \cite{lozano03} to infer the following:
%Applying (\ref{eq:RicTr}) in (\ref{eq:CDDotExpr}) we obtain
%\begin{align} \label{eq:C0ddotMMSERic}
%\ddot{C}^{\rm mmse}(0) = - \frac{\log_2(e) N_r}{(K+1)^2 N_t} \left( K^2 ( 2  N_t - 1)N_r + (2K+1) ( 2 N_t + N_r - 1 ) %\right)
%\end{align}
%which, from the definition (\ref{eq:S0Defn}), yields

As also observed previously for the high SNR regime, we see that the impact of line-of-sight in the low SNR regime is only through the relative strength of the specular component (ie.\ through the Rician $K$-factor).  Moreover, the wideband slope (\ref{eq:S0RicClosedForm}) is a monotonically decreasing function of $K$, implying that line-of-sight has a damaging effect on the achievable sum rate of MIMO MMSE receivers in the low SNR regime. Comparing this result with the corresponding wideband slope for MIMO with optimal receivers, given for Rician fading by \cite{lozano03}
\begin{align} \label{eq:S0RicClosedFormopt}
S_0^{\rm opt} = \frac{ 2 (K + 1)^2 }{ K^2 + (2K+1) \frac{ N_t + N_r  }{N_r N_t} } \;
\end{align}
we obtain the interesting relationship
\begin{align}
\frac{S_0^{\rm mmse}}{ S_0^{\rm opt} } = \frac{ \varphi ( K, N_r, N_t ) }{ \varphi ( K, N_r, 2 N_t - 1 ) }
\end{align}
with
\begin{align}
\varphi ( K, m, n ) = K^2 m n + (2K+1) ( m + n ) \; .
\end{align}
This ratio, again, is a decreasing function of $K$, satisfying
\begin{align}
\frac{ N_t }{ 2 N_t - 1 } \;  \leq \;  \frac{S_0^{\rm mmse}}{ S_0^{\rm opt} } \; \leq \; \frac{ N_t + N_r}{ 2 N_t + N_r - 1}
\end{align}
where the lower bound is approached as $K \to \infty$, corresponding to the purely deterministic channel scenario, and the upper bound is approached as $K \to 0$, corresponding to Rayleigh fading.

%\begin{align}
%\frac{S_0^{\rm mmse}}{ S_0^{\rm opt} } = \frac{ K^2 N_r N_t + (2K+1) ( N_t + N_r ) }{ K^2 N_r (2 N_t - 1) + (2K+1) ( 2 %N_t + N_r - 1 ) } \; .
%\end{align}
%Upon careful inspection, we can observe the following interesting relationship
%\begin{align}
%\frac{S_0^{\rm mmse}}{ S_0^{\rm opt} } = \frac{ \varphi ( K, N_r, N_t ) }{ \varphi ( K, N_r, 2 N_t - 1 ) }
%\end{align}

\section{Concluding Remarks} \label{sec:conclusions}

We presented a new analytic framework for investigating the achievable sum rate of MIMO systems employing MMSE receivers, revealing a simple but powerful connection with the ergodic MIMO mutual information achieved with optimal receivers. This framework allowed us to directly exploit existing MIMO results in the literature, thereby circumventing the major challenges entailed with explicitly characterizing the SINR distribution at the MMSE output. To demonstrate the utility of the framework, we presented particularizations for uncorrelated and correlated Rayleigh fading, and uncorrelated Rician fading channels, yielding new exact closed-form expressions for the MMSE achievable sum rate as well as simplified expressions for the high and low SNR regimes.  Through these expressions, we obtained key analytical insights into the effect of the various system and channel parameters under practical fading conditions.  For example, we demonstrated that at both high and low SNR, the MMSE achievable sum rate is reduced by either spatial correlation or line-of-sight. At high SNR, this rate reduction is manifested as an increased power offset, whereas at low SNR, through a reduced wideband slope.  Moreover, at both high and low SNRs, the rate loss due to spatial correlation or line-of-sight was shown to be more significant for MMSE receivers than for optimal receivers.  We also demonstrated that the effect of line-of-sight on the MMSE achievable sum rate was dependent on the relative strength of the specular component, but not the geometry of such component.

We would like to stress that although the main focus of the paper was on single-user MIMO systems with MMSE receivers, many of the results apply almost verbatim to multi-user scenarios; in particular, the analysis of multiple access channels with MMSE receivers, and MIMO broadcast channels with either MMSE-based transmit precoding (see, eg.\ \cite{Jindal_07_TrIT}) or MMSE reception \cite{louie08_Globe}.  Moreover, the proposed framework extends to many other scenarios beyond those explicitly studied in this paper.  These include, for example, single-user MIMO systems operating in the presence of interference \cite{Chiani_06_Int,Blum_02_02,Blum_03_01} and amplify-and-forward relaying systems \cite{jin07_IT}.

\begin{appendices}

\section{Proof of Theorem \ref{th:MMSESumCapacity_Main}}  \label{ap:MMSESumCapacity_Main}

We start by substituting (\ref{eq:SINRk}) into (\ref{eq:MMSEMI}),
and using\footnote{Note that this property has also been used in relation to linear MIMO receivers in \cite{gore02}.} \cite{Horn90}
\begin{align} \label{eq:DetIdent}
\left[ \mathbf{Z}^{-1} \right]_{i,i} = \frac{ \det \mathbf{Z}}{ \det \mathbf{Z}^{ii} }
\end{align}
where $\mathbf{Z}^{ii}$ is the $(i,i)$th minor of the matrix
$\mathbf{Z}$, to yield
\begin{align} \label{eq:MMSE_Capacity_General}
I^{\rm mmse} ({\rm snr}) &= \sum_{i=1}^{N_t} E_{\mathbf{H}} \left[ \log_2 \left(
\frac{\det \left( \mathbf{I}_{N_t} + \frac{{\rm snr}}{N_t}
\mathbf{H}^\dagger \mathbf{H} \right) }{ \det \left( \big(
\mathbf{I}_{N_t} + \frac{{\rm snr}}{N_t} \mathbf{H}^\dagger
\mathbf{H} \big)^{ii} \right)  }  \right) \right] \nonumber \\
&= N_t E_{\mathbf{H}} \left[ \log_2  \det \left( \mathbf{I}_{N_t} +
\frac{{\rm snr}}{N_t} \mathbf{H}^\dagger \mathbf{H} \right) \right] - \sum_{i=1}^{N_t} E_{\mathbf{H}} \left[ \log_2 \det \left( \mathbf{I}_{N_t - 1} + \frac{{\rm snr}}{N_t} \left(
\mathbf{H}^\dagger \mathbf{H} \right)^{ii} \right)   \right] \; .
\end{align}
Noting that $\left(\mathbf{H}^\dagger \mathbf{H} \right)^{ii} = \mathbf{H}_i^\dagger \mathbf{H}_i$, the result follows from (\ref{eq:CoptDefn_Exp}).

\section{Proof of Theorem \ref{th:MMSESumCapacity_LowSNR}}  \label{ap:MMSESumCapacity_LowSNR}

From (\ref{eq:MMSESumCapacity_Main}), and noting that
\begin{align}
&\frac{ {\rm d}}{ {\rm d} x} \ln \det \left( \mathbf{I} + x \mathbf{A} \right) \biggr|_{x = 0} = {\rm tr} \left( \mathbf{A} \right) \; , \nonumber \\ &\frac{ {\rm d}^2}{ {\rm d}^2 x} \ln \det \left( \mathbf{I} + x \mathbf{A} \right) \biggr|_{x = 0} = - {\rm tr} \left( \mathbf{A}^2 \right)
\end{align}
we can obtain
\begin{align} \label{eq:CdotExpr}
\dot{I}^{\rm mmse}(0) = \log_2e \biggl( E_{\mathbf{H}} \left[ {\rm tr}\left(\mathbf{H} \mathbf{H}^\dagger \right) \right]  - \frac{1}{N_t} \sum_{i = 1}^{N_t} E_{\mathbf{H}_i} \left[ {\rm tr}\left(\mathbf{H}_i \mathbf{H}_i^\dagger \right) \right] \biggr)
\end{align}
and
\begin{align} \label{eq:CDDotExpr}
\ddot{I}^{\rm mmse}(0) & = - \frac{\log_2e}{N_t^2} \left( N_t E_{\mathbf{H}} \left[ {\rm tr}\left( [\mathbf{H} \mathbf{H}^\dagger]^2 \right) \right] - \sum_{i = 1}^{N_t} E_{\mathbf{H}_i} \left[ {\rm tr}\left( [ \mathbf{H}_i \mathbf{H}_i^\dagger ]^2 \right) \right] \right) \nonumber \\
&= - \log_2e \left( N_t N_r \zeta( \mathbf{H} \mathbf{H}^\dagger )  - N_r \left( \frac{N_t-1}{N_t} \right)^2  \sum_{i = 1}^{N_t} \zeta ( \mathbf{H}_i \mathbf{H}_i^\dagger )  \right) \; .
\end{align}
Due to the channel power normalization (\ref{eq:ChannelNorm}), we have
\begin{align}
E_{\mathbf{H}} \left[ {\rm tr}\left(\mathbf{H} \mathbf{H}^\dagger \right) \right] = N_r N_t , \; \;  E_{\mathbf{H}_i} \left[ {\rm tr}\left(\mathbf{H}_i \mathbf{H}_i^\dagger \right) \right] = N_r ( N_t - 1)
\end{align}
regardless of the specific channel statistics.  As such, (\ref{eq:CdotExpr}) evaluates to
\begin{align} \label{eq:C0dotMMSE}
\dot{I}^{\rm mmse}(0) = \frac{N_r}{\ln 2}
\end{align}
which, from (\ref{eq:EbNoDefn}), yields (\ref{eq:EbNoMMSE}). The result (\ref{eq:WideSlope_MMSE}) is obtained by substituting (\ref{eq:C0dotMMSE}) and (\ref{eq:CDDotExpr}) into (\ref{eq:S0Defn}) and simplifying.

\section{Proof of Lemma \ref{le:CorrRayleigh}}
\label{app:CorrRayleigh}

The MIMO mutual information with isotropic inputs is given by
\begin{align} \label{eq:Cap}
I^{\rm opt} (N_r, N_t, {\rm snr}) &= E_{\mathbf{H}} \left[ \log_2  \det \left( \mathbf{I}_{N_t} +
\frac{{\rm snr}}{N_t} \mathbf{H}^\dagger \mathbf{H} \right) \right] \nonumber \\
&= n \int_0^\infty \log_2 \left( 1 + \frac{{\rm snr}}{N_t} \lambda \right) f_\lambda (\lambda) {\rm d}\lambda
\end{align}
where $\lambda$ is an unordered non-zero eigenvalue of $\mathbf{H}^\dagger \mathbf{H}$, with probability density function (p.d.f.) $f_\lambda(\cdot)$.  Recently, the following unified expression (applying for arbitrary $N_r$ and $N_t$) was presented for this p.d.f.\ \cite{jin07_IT}
\begin{align} \label{eq:eigValPDF}
f_\lambda (\lambda) =
\frac{1}{{n\prod\nolimits_{\ell < k}^q {\left( {\beta _k  - \beta _\ell }
\right)} }}\sum\limits_{k = q - n + 1}^q \det {\bf{D}}_k ,
\end{align}
where ${{\bf{D}}_k }$ is a $q \times q$ matrix with entries
\begin{align}
\left\{ {{\bf{D}}_k } \right\}_{s,t}  = \left\{
{\begin{array}{*{20}c}
   {\beta _s^{t - 1} }, & {t \ne k,}  \\
   {\frac{{\lambda ^{p  - q
+ k - 1} }}{{\Gamma \left( {p  - q + k} \right)}}} {e^{ - \lambda /\beta _s } \beta _s^{q - p  - 1} }, & {t = k.}  \\
\end{array}} \right.
\;
\end{align}
The result follows by substituting (\ref{eq:eigValPDF}) into (\ref{eq:Cap}) and integrating using an identity from \cite{alfano04_2}.

\end{appendices}

%\bibliographystyle{IEEEtran}
%\bibliography{IEEEabrv,References_Matt}

\end{document}